\documentclass{article}

\usepackage{amsmath,amsfonts,bm}

\def\eqref#1{equation~\ref{#1}}

\def\1{\bm{1}}

\def\vx{{\bm{x}}}

\def\vz{{\bm{z}}}

\def\mH{{\bm{H}}}

\def\mM{{\bm{M}}}

\def\mX{{\bm{X}}}

\def\mZ{{\bm{Z}}}

\DeclareMathAlphabet{\mathsfit}{\encodingdefault}{\sfdefault}{m}{sl}
\SetMathAlphabet{\mathsfit}{bold}{\encodingdefault}{\sfdefault}{bx}{n}

\def\gG{{\mathcal{G}}}

\def\gN{{\mathcal{N}}}

\newcommand{\E}{\mathbb{E}}

\newcommand{\R}{\mathbb{R}}

\def\X{{\mathcal{X}}}

\def\E{{\mathrm{E}}}
\def\OO{{\mathrm{O}}}
\def\SO{{\mathrm{SO}}}

\usepackage{microtype}
\usepackage{graphicx}
\usepackage{subfigure}
\usepackage{booktabs} 
\usepackage{caption}
\usepackage{subcaption}

\usepackage{hyperref}
\newcommand{\cmark}{\ding{51}} 
\newcommand{\xmark}{\ding{55}} 
\usepackage{amssymb} 
\usepackage{pifont} 

\usepackage[accepted]{icml2024}

\usepackage{amsmath}
\usepackage{amssymb}
\usepackage{mathtools}
\usepackage{amsthm}
\usepackage{xspace}
\usepackage{multirow}

\usepackage[capitalize,noabbrev]{cleveref}

\theoremstyle{plain}
\newtheorem{theorem}{Theorem}[section]

\newtheorem{lemma}[theorem]{Lemma}

\theoremstyle{definition}
\newtheorem{definition}[theorem]{Definition}

\theoremstyle{remark}

\usepackage[textsize=tiny]{todonotes}
\usepackage{enumitem}
\usepackage{hyperref}

\usepackage{xcolor}

\newcommand{\model}{SymGNN\xspace}

\begin{document}

\twocolumn[
\icmltitle{Predicting and Interpreting Energy Barriers of Metallic Glasses with Graph Neural Networks}



\icmlsetsymbol{equal}{*}

\begin{icmlauthorlist}
\icmlauthor{Haoyu Li}{equal,1}
\icmlauthor{Shichang Zhang}{equal,1}
\icmlauthor{Longwen Tang}{1}
\icmlauthor{Mathieu Bauchy}{1}
\icmlauthor{Yizhou Sun}{1}
\end{icmlauthorlist}

\icmlaffiliation{1}{University of California, Los Angeles, CA, USA}

\icmlcorrespondingauthor{Haoyu Li}{haoyuli02@ucla.edu}
\icmlcorrespondingauthor{Shichang Zhang}{shichang@cs.ucla.edu}

\icmlkeywords{Machine Learning, ICML}

\vskip 0.3in
]



\printAffiliationsAndNotice{\icmlEqualContribution} 

\begin{abstract}
Metallic Glasses (MGs) are widely used materials that are stronger than steel while being shapeable as plastic. While understanding the structure-property relationship of MGs remains a challenge in materials science, studying their energy barriers (EBs) as an intermediary step shows promise. In this work, we utilize Graph Neural Networks (GNNs) to model MGs and study EBs. We contribute a new dataset for EB prediction and a novel Symmetrized GNN (SymGNN) model that is E(3)-invariant in expectation. SymGNN handles invariance by aggregating over orthogonal transformations of the graph structure. 
When applied to EB prediction, SymGNN are more accurate than molecular dynamics (MD) local-sampling methods and other machine-learning models. Compared to precise MD simulations, SymGNN reduces the inference time on new MGs from roughly \textbf{41 days} to \textbf{less than one second}. We apply explanation algorithms to reveal the relationship between structures and EBs. 
The structures that we identify through explanations match the medium-range order (MRO) hypothesis and possess unique topological properties. Our work enables effective prediction and interpretation of MG EBs, bolstering material science research.
\footnote{Code for this project is available at \url{https://github.com/haoyuli02/SymGNN}}


\end{abstract}

\section{Introduction}\label{sec:introduction}
\begin{figure}[t]
\begin{center}
\includegraphics[clip, width=\columnwidth]{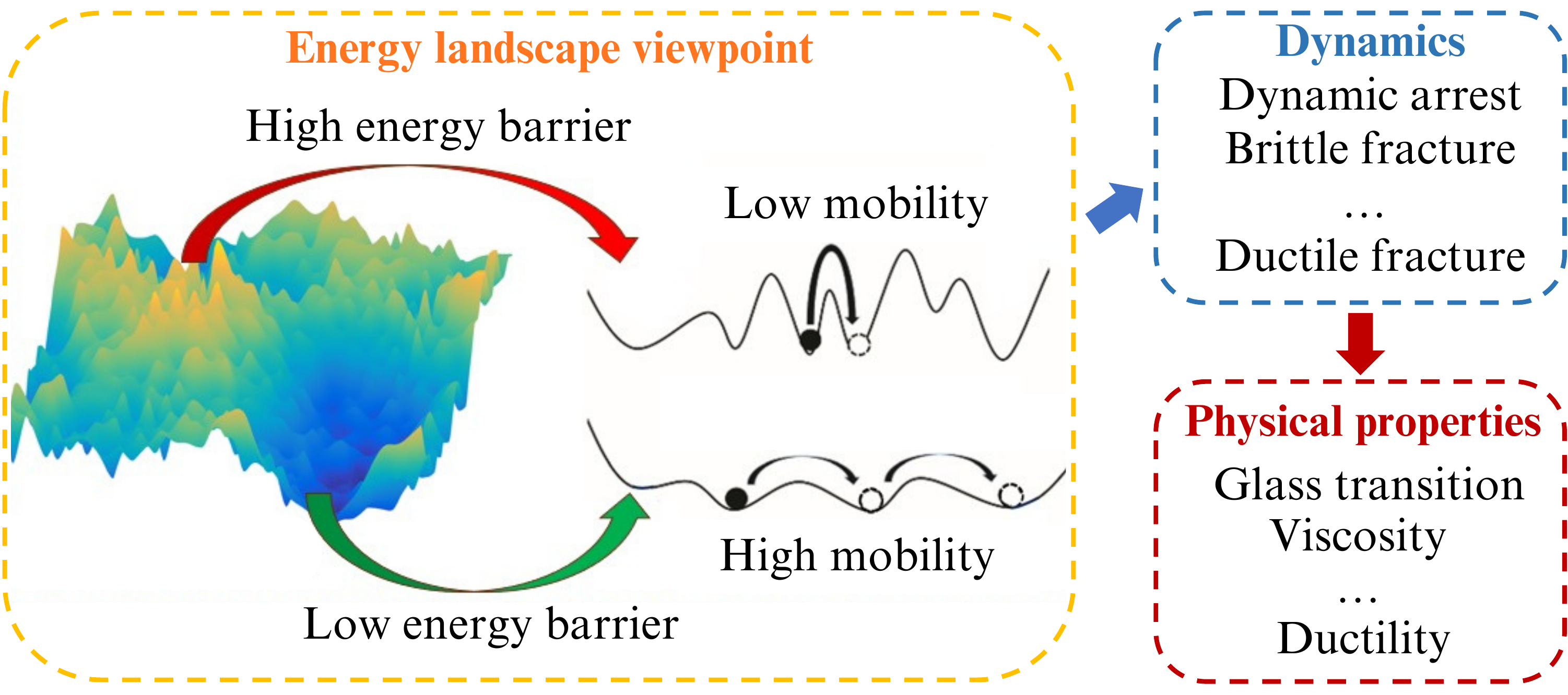}
\end{center}
\caption{EBs represent mobility, which can further influence the MG dynamics and their physical properties.}
\label{fig:energy_barrier}
\end{figure}

Metallic glasses (MGs) combine good properties of metals and plastics in one material, making them stronger than steel while being shapeable as plastic~\cite{SCHROERS201114}. Their extensive applications span various industries including aerospace, sports equipment, luxury goods, biomedical devices, and many more~\cite{trexler2010mechanical}. The unique properties of MGs lie in their non-crystalline amorphous atomic structure, which sets them apart from the crystalline structure found in traditional metals~\cite{trexler2010mechanical, bansal2013handbook}. Despite extensive research on MGs, the details of their structure-property relationship are still not well understood~\cite{starr2002we, ding2014soft, patinet2016connecting, cao2018structural}. 

One promising approach for studying the structure-property relationship of MGs is through a special property called \textit{Energy Barrier (EB)}. EBs describe the local roughness of the energy landscape by comparing the average energy difference around an atom's local neighbors. Many studies have shown that understanding EBs can act as an important intermediary step for studying the MG physical properties~\cite{debenedetti2001supercooled, yu2012correlation, tang2021energy}. As shown in Figure~\ref{fig:energy_barrier}, EBs represent mobility, which can influence the MG dynamics and further their physical properties like glass transition and ductility~\cite{berthier2011theoretical, kirchner2022beyond}. However, the precise simulation of EBs is challenging and often requires time-consuming computation~\cite{barkema1996event, mousseau2012activation, jay2022activation}. For example, even with a high-performance computing (HPC) cluster and the advanced Activation-Relaxation Technique nouveau (ARTn)~\cite{cances2009some}, calculating EBs for an MG system with 3,000 atoms can take 41 days.

Given the usefulness and computational difficulty of EBs, we explore machine learning (ML) approaches to efficiently predict them from MG atomic structures. Similar to recent ML investigations on glassy systems~\cite{2020NatPh..16..448B, reiser2022graph}, we phrase the EB prediction problem as a graph ML problem and solve it using Graph Neural Networks (GNNs). Under this formalization, atoms become nodes in a graph, and edges are constructed between nearby nodes to represent the atomic structure. Atom types are used as node features. Displacement vectors constructed from 3D node coordinates are used as edge features. Then EB prediction becomes a node regression task on graphs. 

We simulate MG systems and employ ARTn to calculate some EBs as training labels. Given the challenge of collecting data, a more data-efficient model with a stronger inductive bias is desired. In particular, the EB prediction problem exhibits E(3)-invariance, i.e., invariance to graph structure transformations including translations, rotations, reflections, and their combinations. We aim for a GNN that can handle such invariance, but general message-passing-based GNNs like GCN~\cite{kipf2017semisupervised} cannot. Some specially designed models are E(3)-invariant~\cite{schutt2017schnet, lu2019molecular, gasteiger2020directional, tholke2022torchmd, liao2022equiformer, batatia2022mace, batzner20223}, but, to the best of our knowledge, none of the existing methods can achieve \textit{invariance}, \textit{expressiveness}, and \textit{scalability} at the same time as we show in Table~\ref{tab:model_comparison}.



\begin{figure}[t]
\begin{center}
\includegraphics[clip, width=\columnwidth]{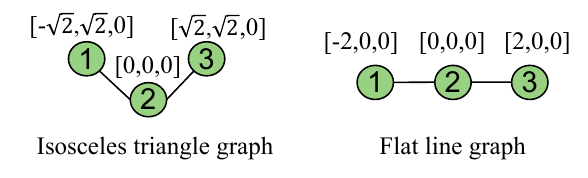}
\end{center}
\caption{Example graphs demonstrating model expressiveness. SchNet cannot distinguish the embeddings of node 1 in these two graphs but \model can.}
\label{fig:expressive_example}
\end{figure}

To achieve an invariant model that is both expressive and scalable, we propose a simple but effective \textit{Symmetrized GNN (\model)}, which is E(3)-invariant in expectation. \model achieves E(3) invariance by introducing a symmetrization module to aggregate embeddings produced under different orthogonal transformations of the graph structure. 
It is expressive as there is no higher-order information loss caused by ``feature scalarization'' as in models like SchNet~\cite{schutt2017schnet}. For example, for the two graphs in Figure~\ref{fig:expressive_example}, SchNet will pass the same message to the central node of an isosceles triangle as if the middle node of a flat line, but SymGNN will easily distinguish between these two configurations (details in Appendix~\ref{app:expressiveness}). Also, SymGNN does not involve complex equivariant calculation, so it is much more scalable than methods relying on equivariant feature extraction~\cite{liao2022equiformer, batatia2022mace, batzner20223}. In our experiments, we demonstrate that when applied to MG graphs, \model outperforms a variety of widely used GNNs.




\def\rot{\rotatebox{20}}
\begin{table}[t]
\centering
\small 
\begin{tabular}{l|ccc}
\hline
Methods & \rot{Invariance} & \rot{Expressiveness} & \rot{Scalability}   \\
\midrule
GCN & \xmark & \cmark & \cmark    \\
EGNN & ? & \cmark & \cmark   \\
SchNet & \cmark & \xmark & \cmark   \\
MGCN & \cmark & \xmark & \cmark \\
FAENet & \cmark & \xmark & \cmark \\
DimeNet & \cmark & \cmark & \xmark   \\
Torch-MD Net & \cmark & \cmark & \xmark \\
Equiformer & \cmark & \cmark & \xmark    \\
\textbf{SymGNN} (ours) & E & \cmark & \cmark    \\
\bottomrule
\end{tabular}
\caption{Comparison of different methods. \cmark \space means the model performs well evaluated by the corresponding category, \xmark \space means not, and ? means ``unclear'', or ``possible after non-trivial
extensions''. The E indicates that the model satisfies the property in expectation.}
\label{tab:model_comparison}
\end{table}



Moreover, to better understand the EB prediction and benefit MG research, we also generate explanations along with the model prediction. Our proposed explanation method extends GNNExplainer~\cite{ying2019gnnexplainer} to the node regression task to generate edge-based structure explanations. It helps us to identify and visualize the importance of each edge in predicting an EB. We also show that the generated explanations match the medium-range order (MRO) hypothesis of MGs and possess unique topological properties that correlate with the optimal-volume cycles in a persistent diagram~\cite{obayashi2018volume}. Our findings provide further insights into the understanding of EBs, and our explanations can potentially benefit new scientific discoveries. We summarize our contributions as the following:

\begin{enumerate}
    \item We formulate a material science research problem of predicting MG EBs as an ML problem of node regression on graphs.
    \item We collect MG data for ML research, with precisely simulated EBs using ARTn.
    \item We propose a simple but effective \model model that exhibits E(3)-invariance in expectation and predicts MG EBs accurately and fast.
    \item We generate explanations for EB predictions that match the MRO hypothesis, express unique topological properties, and provide insights for scientific discoveries.
\end{enumerate}



\section{Related work}\label{sec:related}
\subsection{ML in Material Science}

The application of ML in materials science has seen significant advancements recently, with various models to tackle different aspects of material science problems. Among these, GNNs have emerged as a powerful tool for representing and analyzing materials at the atomic level, owing to their ability to capture the complex relationships and interactions between atoms in a material. For example, estimating the propensity of individual atoms \cite{2020NatPh..16..448B}, potential energy exhibited by a system of atoms \cite{schutt2017schnet}. In these settings, inductive bias of equivariance and invariance often plays a key role in the generalizability of network. For example, in our problem the EB only depends on local molecule configuration and thus are invariant on translation, rotation, and reflection of graphs. To incorporate this physical inductive bias, various invariant and equivariant GNNs have been proposed. Invariant GNNs often restrict graph features to be rotationally invariant, such as edge distances and angles, or reducing the inputs by projecting it onto PCA frames~\cite{schutt2017schnet, gasteiger2020directional, gasteiger2020fast, duval2023faenet}, whereas equivariant networks are proposed to leverage tensorial transformation that can extract equivariant node features \cite{schutt2021equivariant, liao2022equiformer, batatia2022mace, batzner20223, tholke2022torchmd}.



Several evaluation benchmarks for equivariant ML models on molecular dynamics~\cite{bihani2024egraffbench} and solid-state materials systems~\cite{choudhary2020joint,lee2023matsciml} have been proposed. EGRaffBench~\cite{bihani2024egraffbench} provides insights into the performance of various equivariant models in predicting forces in molecular systems, highlighting their potential in simulating atomistic interactions. JARVIS~\cite{choudhary2020joint} and MatSciML~\cite{lee2023matsciml} benchmark ML models for solid-state materials systems and demonstrate the potential of ML models, including various GNNs, in predicting properties and behaviors of solid-state materials, thereby aiding in materials design and discovery. Furthermore, a recent overview~\cite{duval2023hitchhiker} of geometric GNNs for 3D atomic systems offers valuable insights into the development and application of GNNs in materials science, emphasizing the importance of geometric considerations in modeling atomic systems.


\subsection{MGs and EBs}
Understanding the relationship between the atomic structure and physical properties of MGs is one of the greatest challenges for both material science and condensed matter physics~\cite{falk2011deformation,sun2015fracture,nicolas2018deformation}. However, the structure-property relationship of MGs is often challenging to characterize directly due to the complexity of the physical properties~\cite{cubuk2017structure, 2020NatPh..16..448B}. EBs describe the local roughness of the energy landscape by comparing the average energy difference around an atom's local neighbors. They are influential in MG dynamics and their physical properties~\cite{berthier2011theoretical, kirchner2022beyond}, for example, the degree of ductility during fracture~\cite{tang2021energy}. Therefore, EBs can act as an important intermediary step when predicting the physical properties with the atomic structures as inputs~\cite{debenedetti2001supercooled, yu2012correlation, wang2020predicting, tang2021energy}. ML methods have been applied to investigate the relationship between the atomic structures and physical properties in MG~\cite{2020NatPh..16..448B}. For EBs in particular, \cite{wang2020predicting} explored using XGBoost to
classify nodes with the highest 5 percent activation energy. Our work furthers the investigation of \cite{wang2020predicting} by leveraging the natural graph structure using GNNs to perform a regression for EBs and generating insightful explanations.

\section{Problem Setup and Preliminaries}\label{sec:prelim}
\subsection{EB Prediction with GNNs} \label{subsec:problem_formulation}

The problem of predicting EBs of MGs can be formalized as a node regression problem on graphs. Under this formulation, atoms become nodes in a graph, and edges are constructed between nearby nodes. The MG data thus becomes a graph with $n$ nodes and $m$ edges. We represent the graph structure with $G$, which indicates all the edges and is normally represented in the form of an adjacency matrix. The node features are the atom types, which we represent with $\mZ = \{\vz^1, \vz^2, \dots, \vz^n\}$. The edge features are the displacement vectors constructed from 3D node coordinates, which we represent with $\mX = \{\vx^1, \vx^2, \dots, \vx^m\ | \vx^i \in \R^3 \}$. The regression task is to predict the EB label $y \in \R$ of each node with the graph structure and features as inputs, i.e., a model that maximizes $P(y |\, G, \mZ, \mX)$. We further break down the prediction process into two steps. The first step encodes node and edge features to embeddings $\mH$. The second step predicts $y$ with $G$ and $\mH$ as inputs. The objective to maximize becomes the following, 
\begin{equation} \label{eq:problem}
P(y |\, G, \mZ, \mX) = \int_{\mH}P(y |\, G, \mH) P(\mH |\, G, \mZ, \mX) d\mH
\end{equation}
We solve this problem with the state-of-the-art graph ML models - GNNs.

\subsection{Orthogonality and Invariance} \label{subsec:invariance}
EB is invariant to Euclidean transformations of the atomic graph structure, for example, rotations, reflections, and translations, because it is the average energy needed for a node to hop between its current and nearby energy subbasins. Given that the graph is described with displacement vectors of relative positions, translations will be canceled, and the invariance to Euclidean transformations can be reduced to the invariance to orthogonal transformations~\cite{hall2013lie}, which is defined as the following,

\begin{definition}[Orthogonal Transformation]\label{def:ortho_trans}
A linear transformation $T: \mathbb{R}^d \to \mathbb{R}^d$ is called an \textit{orthogonal transformation} if it preserves the inner produce $\langle \cdot, \cdot \rangle$ on $\R^d$, i.e., $\forall \vx^1, \vx^2 \in \mathbb{R}^d$, $\langle T(\vx^1), T(\vx^2) \rangle = \langle \vx^1, \vx^2 \rangle$. Then, the matrix form of $T$ has $|\det(T)| = 1$. The \textit{orthogonal group} in dimension $d$ is the group of all such orthogonal transformations on $\mathbb{R}^d$ and is denoted as $\OO(d)$.
\end{definition} 

We also state a well-known lemma in group theory~\cite{hall2013lie} and a theorem by Euler~\cite{slabaugh1999computing} for decomposing the orthogonal group and rotations respectively. They will be useful for modeling invariance.

\begin{lemma}[$\OO(3)$ Decomposition] \label{lemma:decompose}
The orthogonal group $\OO(3)$ can be decomposed into rotations and non-rotations. The rotations also form a group denoted as $\SO(3)$, and it contains all transformations $R$ whose matrix forms have $\det(R) = 1$. The non-rotations contain all the reflections and roto-reflections (also called improper rotation) $\tilde R$, whose matrix form have $\det(\tilde R) = -1$. Non-rotations can be denoted as $P \cdot \SO(3)$, with $P$ being any reflection transformation through the origin. 
\end{lemma}



\begin{theorem} (Euler)\label{thm:euler}
Define the rotations around the three coordinate axes $x_1, x_2$, and $x_3$ in $\mathbb{R}^3$ by
$$O_{x_1}(\alpha) = 
\begin{bmatrix}
    1 & 0 & 0 \\
    0 & \cos(\alpha) & -\sin(\alpha) \\
    0 & \sin(\alpha) & \cos(\alpha)
\end{bmatrix}$$
$$O_{x_2}(\beta) = 
\begin{bmatrix}
    \cos(\beta) & 0 & -\sin(\beta) \\
    0 & 1 & 0 \\
    \sin(\beta) & 0 & \cos(\beta)
\end{bmatrix}$$
$$O_{x_3}(\gamma) = 
\begin{bmatrix}
    \cos(\gamma) & -\sin(\gamma) & 0 \\
    \sin(\gamma) & \cos(\gamma) & 0 \\
    0 & 0 & 1
\end{bmatrix}$$
Then any rotation $R \in \SO(3)$ can be written as $R_{\alpha, \beta, \gamma} = O_{x_1}(\alpha)O_{x_2}(\beta)O_{x_3}(\gamma)$ for some angles $[\alpha, \beta, \gamma] \in [-\pi, \pi]^3$. 
These angles are called the Euler angles. 
\end{theorem}

Then we formally introduce the invariant/equivariant transformation.

\begin{definition}[Invariant/Equivariant Transformation]
    Given a group $K$ acts on $\R^d$. A transformation $T: \R^d \to \R^d$ is invariant to $K$ if $T(\vx) = T(k\cdot \vx)$ 
    and equivariant to $K$ if $k\cdot T(\vx) = T(k \cdot \vx)$ for all $k \in K$ and for all $\vx \in \R^d$.
\end{definition}

\subsection{GNNExplainer}
As a representative GNN explanation method, GNNExplainer seeks to explain GNN classifications by selecting an important edge-induced subgraph $G_S$ that minimizes the entropy $H(\cdot)$ of the label $Y$. Since $G_S$ is discrete, GNNExplainer learns a continuous distribution $\gG$ over $G_S$ that gives the minimal expected entropy, where $\gG$ can be implemented as a learnable edge mask $\mM \in \R^{|G|}$ applied on edges of $G$ after a sigmoid function $\sigma$. Mathematically, the optimization objective is 
\begin{equation}\label{eq:gnnexplainer}
    \min_{\gG} \mathbb{E}_{G_S \sim \gG} H(Y|G=G_S) =\min_{\mM } H(Y|G=\sigma(\mM) \odot G)
\end{equation}

\section{Method} \label{sec:methodology}
In this section, we present \model for solving the EB prediction problem we formalized in Section \ref{subsec:problem_formulation}. We first introduce the theory behind the core symmetrization module for capturing $\OO(3)$ invariance in Section \ref{subsec:symmetrization}, then the full \model model in Section \ref{subsec:symgnn}, and finally how we apply explanation algorithms to \model to reveal the connection between the atomic structures and EBs in Section \ref{subsec:explain}.

\subsection{Theory of Symmetrization Over $\OO(3)$} \label{subsec:symmetrization}
Although EB is invariant to Euclidean transformations of the atomic graph structure, most GNNs are not designed to automatically capture such invariance. There are existing GNNs specialized for molecular graphs that can handle such invariance, but they either utilize scalarization that cannot handle higher-order information, or cannot scale up to graphs with thousands of nodes like MGs. We thus propose a symmetrization module that can better capture invariance and efficiently scale up. This section presents the theory behind the symmetrization.

For the node regression problem formalized in Section \ref{subsec:problem_formulation}, $\mX$ only represents one set of displacement vectors under one particular coordinate system. To achieve $\OO(3)$-invariant (and thus $\E(3)$-invariant as explained in Section~\ref{subsec:invariance}) predictions, we propose a symmetrization over all orthogonal transformations of $\mX$, denoted as $\X = \{T(\mX) \, | \, \forall T \in \OO(3) \}$. Under symmetrization, we reformulate the feature encoding step in Equation~\ref{eq:problem}, i.e.,$P(\mH \, | \, G, \mZ, \X)$, as a probability integrated over $\X$, i.e., 
\begin{equation} \label{eq:integrate}
    P(\mH \, | \, G, \mZ, \X) 
    = \int_{T \in \OO(3)} P(\mH \, | \, G, \mZ, T(\mX))P(T) \, dT 
\end{equation}

Notice that a truly $\OO(3)$-invariant model will give the same result for $P(\mH \, | \, G, \mZ, \mX)$ and $P(\mH \, | \, G, \mZ, \X)$. In the new formulation, when maximizing $P(\mH \, | \, G, \mZ, \X)$, the model will learn the desired invariance by foreseeing and aggregating different transformed graphs. To model such an integral, we first define the distribution of $T$ on $\OO(3)$ through the following two lemmas.

\begin{lemma}\label{lemma:non-rotation} 
Any non-rotation $\tilde R \in P \cdot \SO(3)$ can be written as $\tilde R_{\alpha, \beta, \gamma} = - O_{x_1}(\alpha)O_{x_2}(\beta)O_{x_3}(\gamma)$ for some parameters $[\alpha, \beta, \gamma] \in [-\pi, \pi]^3$. 
\end{lemma}

\begin{proof}
    Please refer to Appendix \ref{proof:non-rotation}.
\end{proof}

Intuitively, Theorem~\ref{thm:euler} says that any rotation in 3D can be decomposed into a combination of rotations that rotate only around the $x_1$-axis, $x_2$-axis, and $x_3$-axis and parameterized with the Euler angels. Similarly, Lemma~\ref{lemma:non-rotation} says a similar decomposition and parameterization can be achieved for non-rotations as well. Bring these two results together gives the following lemma for decomposing any orthogonal transformation $T \in \OO(3)$.

\begin{lemma}\label{lemma:ortho} 
Any orthogonal transformation $T \in \OO(3)$ can be written as $T_{\lambda, \alpha, \beta, \gamma} = (-1)^{\lambda} O_{x_1}(\alpha)O_{x_2}(\beta)O_{x_3}(\gamma)$ for some parameters $\lambda \in \{0, 1\}$ and $[\alpha, \beta, \gamma] \in [-\pi, \pi]^3$.
\end{lemma}

\begin{proof}
    Follow from Lemma~\ref{lemma:decompose}, Theorem~\ref{thm:euler}, and Lemma~\ref{lemma:non-rotation}.
\end{proof}

Lemma~\ref{lemma:ortho} allows the integral in Equation~\ref{eq:integrate} to be reduced into an integral over $\lambda$ and $[\alpha, \beta, \gamma]$ in Equation~\ref{eq:decompose_integral}, which is the objective our GNN will model.
\begin{small}
\begin{align}~\label{eq:decompose_integral}
    &P(\mH \, | \, G, \mZ, \X) =  \\
    &\int_{\lambda, \alpha, \beta, \gamma} P(\mH \, | \, G, \mZ, T_{\lambda, \alpha, \beta, \gamma}(\mX))P(T_{\lambda, \alpha, \beta, \gamma}) \, d\lambda d\alpha d\beta d\gamma \nonumber 
\end{align}
\end{small}

\subsection{Symmetrized GNN} \label{subsec:symgnn}
We now present the full SymGNN model with an illustration shown in Figure~\ref{fig:symgnn}. SymGNN consists of two sub-modules. The first is the symmetrization module mentioned above for producing $\OO(3)$-invariant embeddings $\mH$, which we indicate with $\mH = Sym(G, \mZ, \mX)$. The second is a prediction module that takes the symmetrized $\mH$ and $G$ to perform message passing with attention and then node regression.

\begin{figure*}[t]
\begin{center}
\includegraphics[width=\textwidth]{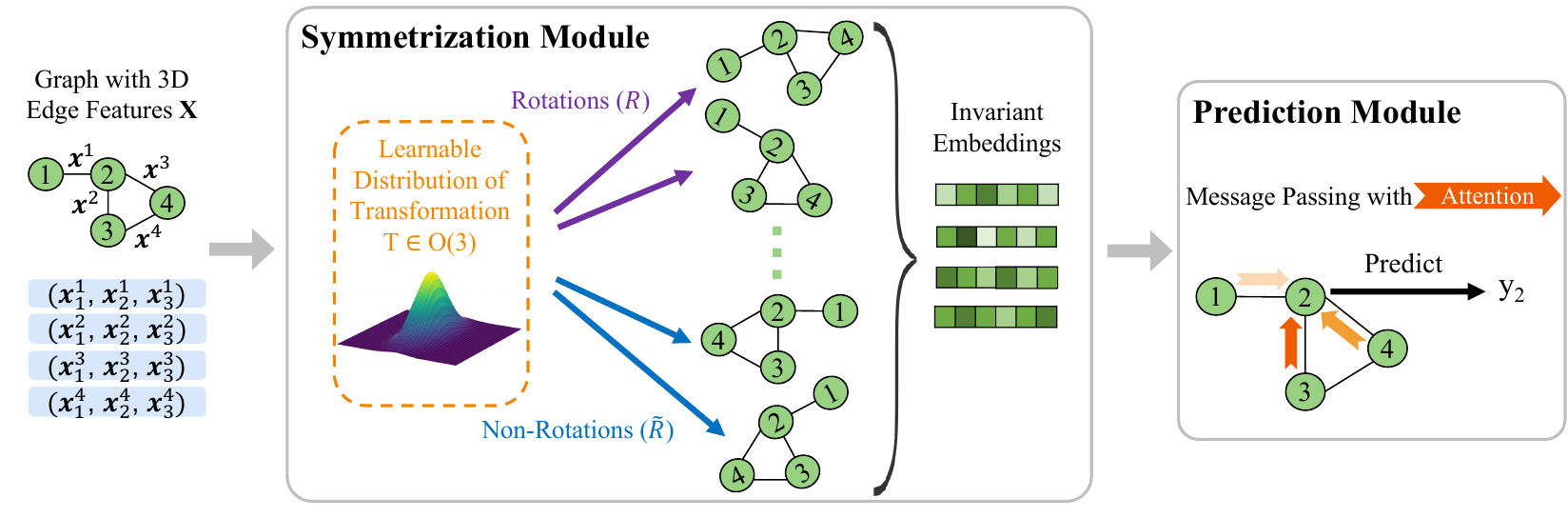}
\end{center}
\caption{Illustration of the \model framework. Given an input graph with node features being atom types and edge feature the relative distance, the symmetrization module of SymGNN aggregates encoding results on various orthogonally transformed graphs sampled from a learnable distribution to achieve $\OO(3)$-invariant in expectation. The invariant embeddings are then passed to message-passing layers with attention to aggregate information and predict label.}
\label{fig:symgnn}
\end{figure*}

The $Sym$ module produces embeddings following the objective in Equation~\ref{eq:decompose_integral} with a learnable encoder $Enc$, i.e., 
\begin{small}
\begin{align} \label{eq:symgnn}
    \mH &= Sym(G, \mZ, \mX) \\
    =& \int_{\lambda, \alpha, \beta, \gamma} Enc(G, \mZ, T_{\lambda, \alpha, \beta, \gamma}(\mX)) P(T_{\lambda, \alpha, \beta, \gamma}) d\lambda d\alpha d\beta d\gamma \nonumber 
\end{align}
\end{small}

However, one challenge is that there are infinitely many $T \in \OO(3)$, which makes the integral intractable. To model such an integral, we generate transformations $T_1, \dots, T_k$ from $\OO(3)$ by sampling $\lambda$ and $[\alpha, \beta, \gamma]$ to approximate the $Sym$ in Equation~\ref{eq:symgnn}, which gives the $Sym_{T_1, \dots, T_k}$ we use in practice.
\begin{equation} \label{eq:symgnn_sample}
    Sym_{T_1, \dots, T_k}(G, \mZ, \mX) 
    = \frac{1}{k} \sum_{i=1}^{k} Enc(G, \mZ, T_{i}(\mX))
\end{equation}

We show that $Sym_{T_1, \dots, T_k}$ is $\OO(3)$-invariant in expectation under assumptions of uniform distributions.

\begin{theorem} \label{thm:invariance}
    Assume $T_1, \dots, T_k$ are random transformations that follow a uniform distribution over all $T \in \OO(3)$. Then, $Sym$ is $\OO(3)$-invariant in expectation in the sense that $\mathbb{E}_{T_1, \dots, T_k} [Sym_{T_1, \dots, T_k}(G, \mZ, \mX)] = \mathbb{E}_{T_1, \dots, T_k}[Sym_{T_1, \dots, T_k}(G, \mZ, T_0\mX)]$ for any $T_0 \in \OO(3)$.
\end{theorem}

\begin{proof}
    Please refer to Appendix \ref{proof:invariance}.
\end{proof}

$Sym$ can learn from a variety of orthogonal transformations and achieve invariance. In practice, we fix $\lambda$ to be $\mathrm{Bern}(0.5)$ to balance rotations and non-rotations uniformly, but we parametrize $\alpha, \beta, \gamma$ with learnable von Mises (Tikhonov) distributions~\cite{mardia1975algorithm} instead of uniform. Learnable distributions help $Sym$ more efficiently sample orthogonal transformation that benefits the prediction. The von Mises parameterization indeed leads to better empirical performance than uniform, since these distributions closely approximate the wrapped normal distribution on $[-\pi, \pi]$.


The second prediction module takes the invariant embeddings $\mH$ produced by $Sym$ to perform message passing and predict $y$. Given the complexity of the prediction problem and to enhance model expressiveness, we also compute attention of edge features and add skip connections during message passing. Specifically, we build up on the Edge Graph Attention Network (EGAT)~\cite{10.1093/bib/bbab371} model to add edge features to the attention calculation in addition to the regular GAT. 

Specifically, after the message from each node is computed, we first calculate an attention score $a_{ij}$ over the edge between nodes $i$ and $j$. Then the representation of node $i$ in the $l+1$-th layer ($e_i^{l+1}$) is constructed as the attention-weighted average of the neighbor representations from the $l$-th layer. We show the formula in the following, where $\sigma$ represents the non-linear activation function and $\gN(i)$ represents the set of neighbors of node $i$.
\begin{align*}
    a_{ij}^l = \frac{\exp{x_{ij}^l}}{\sum_{k\in \gN(i)} \exp{x_{ik}^l}}, \quad
    e_i^{l+1} = \sigma(\sum_{j\in \gN(i)}a_{ij} e_j^l)
\end{align*}

\subsection{Computation Time Analysis}
We provide a theoretical analysis of the time complexity of \model against three other baselines including SchNet, DimeNet, and Equiformer. For SymGNN, the time complexity is $O(kn + nd^2)$, where $k$ denotes the number of sampled orthogonal transformation, $n$ denotes the number of nodes, and $d$ denotes the average number of neighbors a node has. The first $kn$ term comes from the symmetrization of $k$ orthogonal transformations, the second $nd^2$ term comes from the attention operation among a node’s neighbors.
Add these two terms together we essentially get $O(nd^2)$ (empirically we set $k=6$). In comparison, SchNet is faster as it is $O(nd)$. But for DimeNet, since it considers pairwise edge interaction, the time complexity
grows at least as $O(m^2) = O(n^4)$, which makes it prohibitively slow for our graphs with thousands of nodes. For Equiformer, its exact time complexity is unknown to us. It is a transformer-based method where they change attention
to equivariant attention and linear layer to equivariant linear layer using complex tensor operations, which implies that its big-O time complexity is lower bounded by $O(nd^2)$. Empirically, we found that training Equiformer is excessively slow for our large-scale graphs. A list of analyzed time complexity can be found in Table~\ref{tab:time_complexity}. 

\begin{table}[t]
\centering
\caption{Theoretical Time Complexity.}
\vspace{4pt}
{\small 
\setlength\tabcolsep{4pt} 
\begin{tabular}{lcccc}
\midrule
    & SymGNN & SchNet & Equiformer & DimeNet \\ \midrule
Complexity & $O(kn + nd^2)$ & $O(nd)$ & $\Omega(nd^2) $ & $O(n^4)$        \\ \midrule
\end{tabular}
}
\label{tab:time_complexity}
\end{table}

\subsection{Explanations for Structure-EB Relationship} \label{subsec:explain}
ML models have emerged as powerful tools in scientific research, and their utility can extend beyond mere predictions to explanations. This explanatory aspect is crucial because it aligns with the fundamental objective of ML for science: identifying patterns that can elude human analysis and understanding the underlying mechanisms that govern phenomena. 

To make the best use of the SymGNN model and truly bolster the scientific research of MGs, we generate explanations to better reveal the structure-EB relationship. We choose GNNExplainer as a starting point for selecting a subgraph $G_S$ with important edges. Since GNNExplainer was developed for classification problems, the cross-entropy-based objective does not apply to the regression problem of EB prediction. Therefore, we still learn an edge mask $\mM$ on all edges, but modify the objective in Equation~\ref{eq:gnnexplainer} by replacing entropy with mean squared error (MSE) as below, with $f$ representing the \model model.
\begin{align}
    &\min_{\gG} \mathbb{E}_{G_S \sim \gG} \mathrm{MSE}(f(G_S)) 
    =\min_{\mM} \mathrm{MSE}(f(\sigma(\mM) \odot G))
\end{align}
    
This regression explainer considers all edges involved in the prediction of EB for one node and assigns a score to each edge. These scores represent the importance of their corresponding edges for making the prediction. In Section~\ref{sec:explanation}, we demonstrate that the important edges identified by our explainer match the MRO insights mentioned in previous material research and possess unique topological properties.

\section{Experiments} \label{sec:experiment}
We conduct experiments by first constructing an MG dataset with energy barriers simulated by molecular dynamics. Then we apply \model to this dataset and compare its performance with other baseline models. We also perform ablation studies of the symmetrization module to show its effectiveness.

\subsection{Dataset}
\textbf{The proposed Cu64Zr36 dataset.} We employ molecular dynamics to simulate the behavior of a representative Cu64Zr36 MG subjected to shear deformation. The simulated MG system comprises 8000 atoms, generated through the conventional melting-quenching procedure. To evaluate the influence of system size, we also simulate small systems with 3000 atoms. To ensure the statistical robustness of our findings, 9 independent metallic glass samples are generated. To obtain the energy barriers of atoms, we employ the activation-relaxation technique nouveau (ARTn)~\cite{barkema1996event, cances2009some} to calculate the energy barriers. The simulated results are used to construct a dataset consisting of nine graphs. Among them, six graphs are used for training, one graph is for validation, and two graphs are for testing. Each training/validation graph has 8,000 nodes and roughly 260,000 edges, and each test graph has 3,000 nodes and roughly 100,000 edges. In Appendix~\ref{app:data_collect}, we provide a detailed description of the dataset construction process including the simulation temperature, pressure control, cooling rates, ARTn saddle point search algorithms, edge construction threshold, and etc.

\textbf{Cu-Zr MGs from~\cite{wang2020predicting}}
We also tested our method on other metallic glass dataset. We adopted the dataset proposed in one of the previous work~\cite{wang2020predicting}, which includes two additional Cu-Zr type metallic glasses Cu80Zr80 and Cu50Zr50. For each type of material the dataset contains two graphs each with 5000 nodes and around 650000 edges. We picked one as training graph and the other one for testing.

\subsection{Experiment Settings}
\textbf{Baselines:} We evaluate our model against a variety of other ML models including Graph Convolutional Network (GCN)  \cite{kipf2017semisupervised} withe edge features, E(n) Equivariant GNN (EGNN) \cite{satorras2022en} that are designed to handle equivariant features, a non-graph based multi-layer perceptron (MLP) model, and various invariant baselines that are proposed to handle molecular data including SchNet~\cite{schutt2017schnet}, MGCN~\cite{lu2019molecular}, DimeNet~\cite{gasteiger2020directional}, Torch-MD Net~\cite{tholke2022torchmd}, Equiformer~\cite{liao2022equiformer}, and FAENet~\cite{duval2023faenet},. Furthermore, we perform two ablation studies named SymGNN w/o symmetrization where we remove the symmtrization layer and Data Augmentation where we only aggregated the embedding from three fixed orthogonal transformation instead of a learned one. In addition we have include a simple baseline in which we use the absolute length of edge instead of its 3D coordinates as an input edge feature to achieve invariance. We also compared to MD based local sampling approximation that is widely used by material scientists \cite{krishnan2017enthalpy, sastry1998signatures}. 



\textbf{Evaluation:} 
The predicted energy barriers are evaluated by the Pearson product-moment correlation coefficient against the true values following from previous work in material science literature~\cite{2020NatPh..16..448B}. We run each experiment 4 times with different random initializations. On our Cu64Zr36 dataset, we use the validation set to determine the best model and compute the score with the best model on the test set. For the other Cu-Zr MGs dataset, we compute the test accuracy on the final epoch since there is no validation set.

\textbf{Implementation:}
For our proposed Cu64Zr36 dataset, we train 4-layer GNNs for 20,000 epochs using an Amsgrad optimizer \cite{reddi2019convergence} with a learning rate of 0.0001. We adopt an early stopping scheme if the model's prediction score on the validation set did not improve for 1000 epochs. For the other Cu-Zr MGs dataset, we train each model to a fixed number of epochs as there is no validation set. For those models that have smaller scale and faster convergence, i.e., MLP, GCN, EGNN, EGAT, SchNet, and MGCN, we train to 5000 epochs. For SymGNN, we trained the model to 10000 epochs for better convergence. In all the datasets for SymGNN, the distribution over the angles $\alpha, \beta$, and $\gamma$ is parameterized by the von Mises (Tikhonov) distribution.




\subsection{Prediction Results}

\begin{table*}[t]
\caption{Testing scores of the molecular dynamics (MD) method and machine learning (ML) methods. Test results are with the best model on the validation set. Our \model significantly outperforms the MD method and achieves the best among all the ML methods. }
\label{tab:exp_result}
\center
\begin{tabular}{@{}lllll@{}}
\toprule
    & Methods &  Cu64Zr36 & Cu80Zr20 & Cu50Zr50 \\ \midrule
    MD & Local Sampling~\cite{sastry1998signatures} &  $ 0.3614 $  & $ - $ & $ - $\\ \midrule
    \multirow{2}{*}{Non-Invariant ML}
    & MLP & $0.0575 \pm 0.0127$ & 0.0727 $\pm$ 0.0154 & -0.0652 $\pm$ 0.0099\\ 
    & GCN with Edge Features & $0.5123 \pm 0.0507$ & 0.2478 $\pm$ 0.0051 & 0.1395 $\pm$ 0.0068\\ \midrule
    \multirow{6}{*}{Invariant ML}
    & E(n) Equivariant GNN & $0.2588 \pm 0.0077$ & 0.1382 $\pm$ 0.0113 & 0.1381 $\pm$ 0.0098\\
    & EGAT (Edge Length as 1D Feature)& $0.7264 \pm 0.0063$ & 0.5489 $\pm$ 0.0218 & 0.1571 $\pm$ 0.0095\\
    & SchNet & $0.7588 \pm 0.0088$ & 0.2505 $\pm$ 0.0128& 0.1808 $\pm$ 0.0106\\
    & MGCN & $0.7352 \pm 0.0066$ & 0.1793 $\pm$ 0.0133 & 0.1596 $\pm$ 0.0033\\
    & FAENet & $0.6603$ $\pm$ 0.0218 & 0.2947 $\pm$ 0.0171 & 0.2214 $\pm$ 0.0160\\
    \midrule
    \multirow{3}{*}{Ours}
    & SymGNN  & \textbf{0.7859 $\pm$ 0.0056} & \textbf{0.6084 $\pm$ 0.0167} & \textbf{0.5862 $\pm$ 0.0277}\\
    & SymGNN w/o symmetrization & $0.2669 \pm 0.0371$ & 0.2283 $\pm$ 0.0256 & 0.1135 $\pm$ 0.0129\\ 
    & Data Augmentation & $0.6614$ $\pm$ 0.0285 & 0.3304 $\pm$ 0.0201 & 0.2135 $\pm$ 0.0337\\
\bottomrule
\end{tabular}
\end{table*}

We report the results of \model and other baselines in Table~\ref{tab:exp_result}. It can be seen from the table that \model outperforms the baselines by a large amount and exhibits a much stronger generalization power. When we remove the symmetrization module, (i.e. SymGNN w/o symmetrization), the ablated model cannot generalize well, and a similar performance drops is observed when we only aggregates embedding from three fixed orthogonal transformations. This demonstrates the effectiveness of the symmetrization module. Also we observe that models capable of handling invariance can lead to much better result compared to the ones that cannot, which again highlights the importance of symmetrization module in achieving good prediction performance. In addition, we also ran Equiformer, Torch-MD Net, and DimeNet as our baselines. However, we noticed that the training time for these methods are prohibitively long (i.e longer than 2 days) on our dataset.

\begin{table}[h!]
\centering
\caption{Training time comparison for one epoch.}
\vspace{10pt}
\begin{tabular}{lcccc}
\midrule
\textbf{}    & SymGNN & SchNet & Equiformer & DimeNet \\ \midrule
Time & 3 secs          & 1 sec           & 82 mins             & -                \\ \midrule
\end{tabular}
\label{tab:training_time}
\end{table}

\begin{table}[t]
\centering
\caption{Inference time comparison on an MG with 3,000 atoms.}
\vspace{10pt}
\begin{tabular}{lccc} 
\midrule
& SymGNN & ARTn & MD local sampling \\ 
\midrule
Time & 0.26 seconds & 41 days &  150 mins \\ 
\midrule
\end{tabular}
\label{tab:inference_time}
\end{table}

\subsection{Computation Time Analysis}
We notice that SymGNN reaches high performance without dramatically increase both the training cost or the inference cost. Table~\ref{tab:training_time} provides an empirical time comparison for the time needed to train the model for one epoch. We observe that DimeNet would run indefinitely for our larger graph, and the time taken by Equiformer is also prohibitively long. Table~\ref{tab:inference_time} shows a inference time comparison. Compared to traditional MD simulation, our ML-based approach needs much fewer computation resources and is much more efficient. For precise MD simulation with ARTn, the calculation of the energy barrier for each atom takes around 20 minutes in a supercomputer with 16 parallel threads. Therefore, for a MG system that has the size of our test graph, i.e., 3,000 atoms, the computation will take $\frac{20 \times 3000}{60 \times 24} \approx 41$ days. Even for the much faster and less inaccurate local sampling method, the inference time for this MG system can take 150 minutes. In contrast, SymGNN's inference time on the test graph is almost negligible.

\section{Explanation and Analysis}\label{sec:explanation}
\begin{figure}[t]
    \centering
    \includegraphics[width=0.48\textwidth]{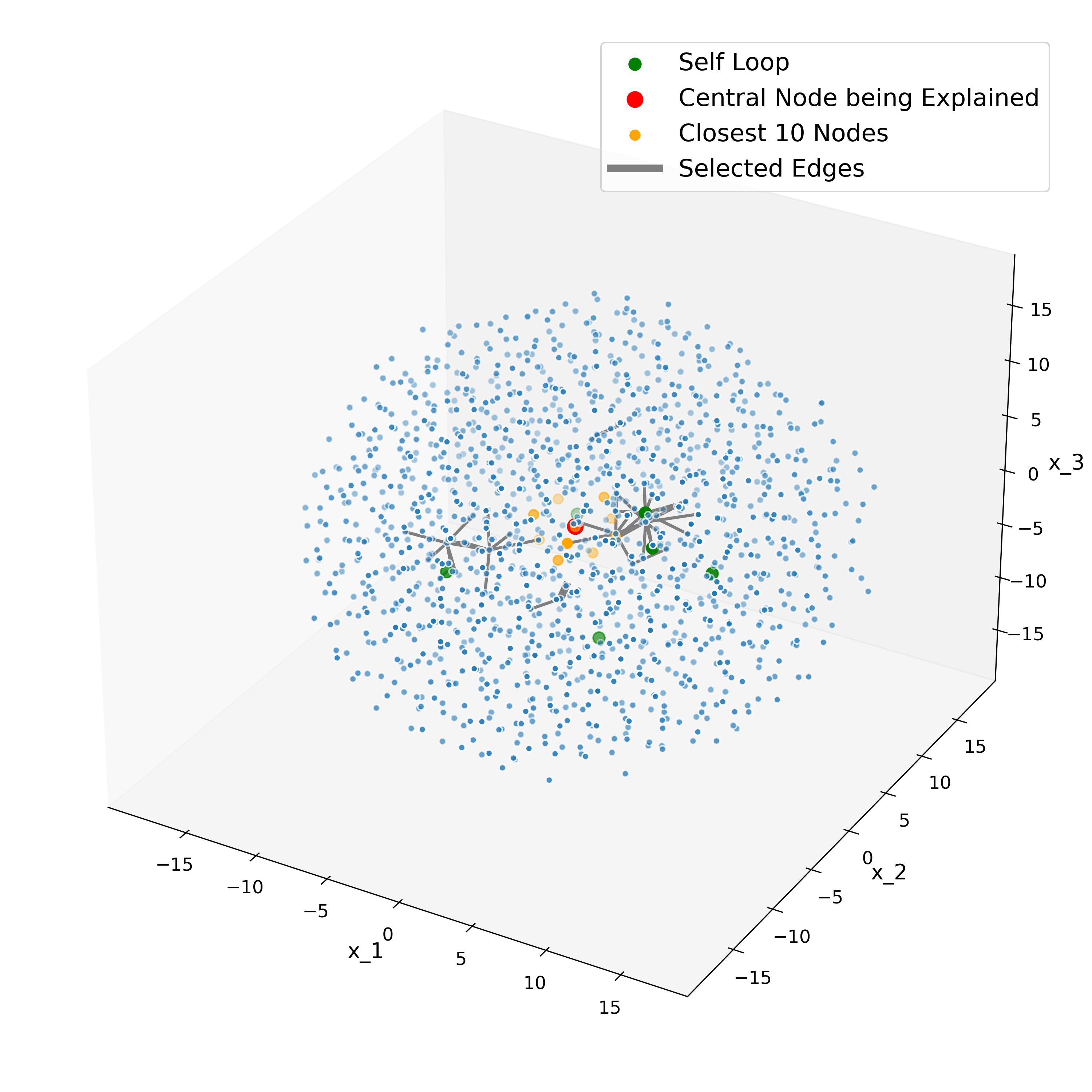}
    \caption{Global explanation visualization. We notice that many of selected edges are of 2 hop neighborhood from the central node.}
    \label{fig:global}
\end{figure}

\begin{figure*}[t]
    \centering
    \begin{minipage}[b]{0.64\textwidth}
        \centering
        \includegraphics[width=\textwidth]{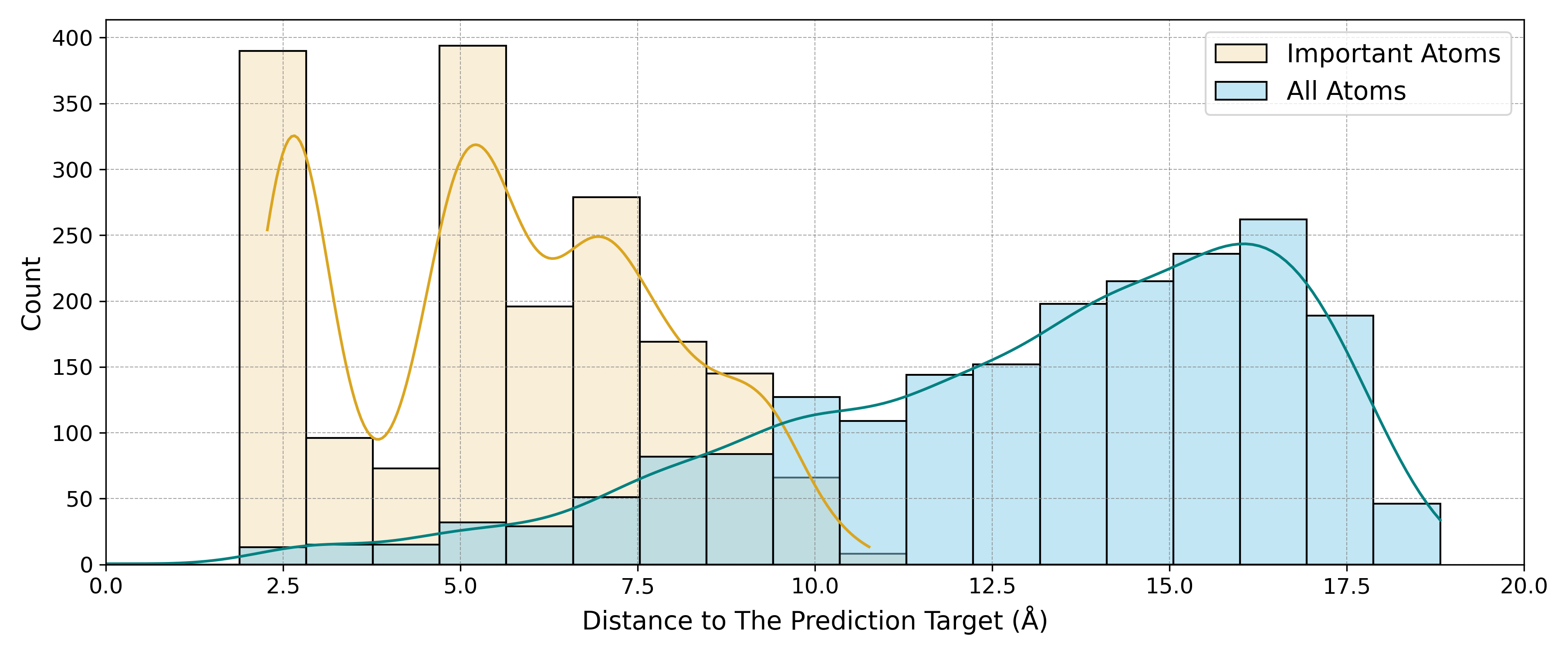}
    \end{minipage}\hfill
    \begin{minipage}[b]{0.36\textwidth}
        \centering
        \includegraphics[width=\linewidth]
        {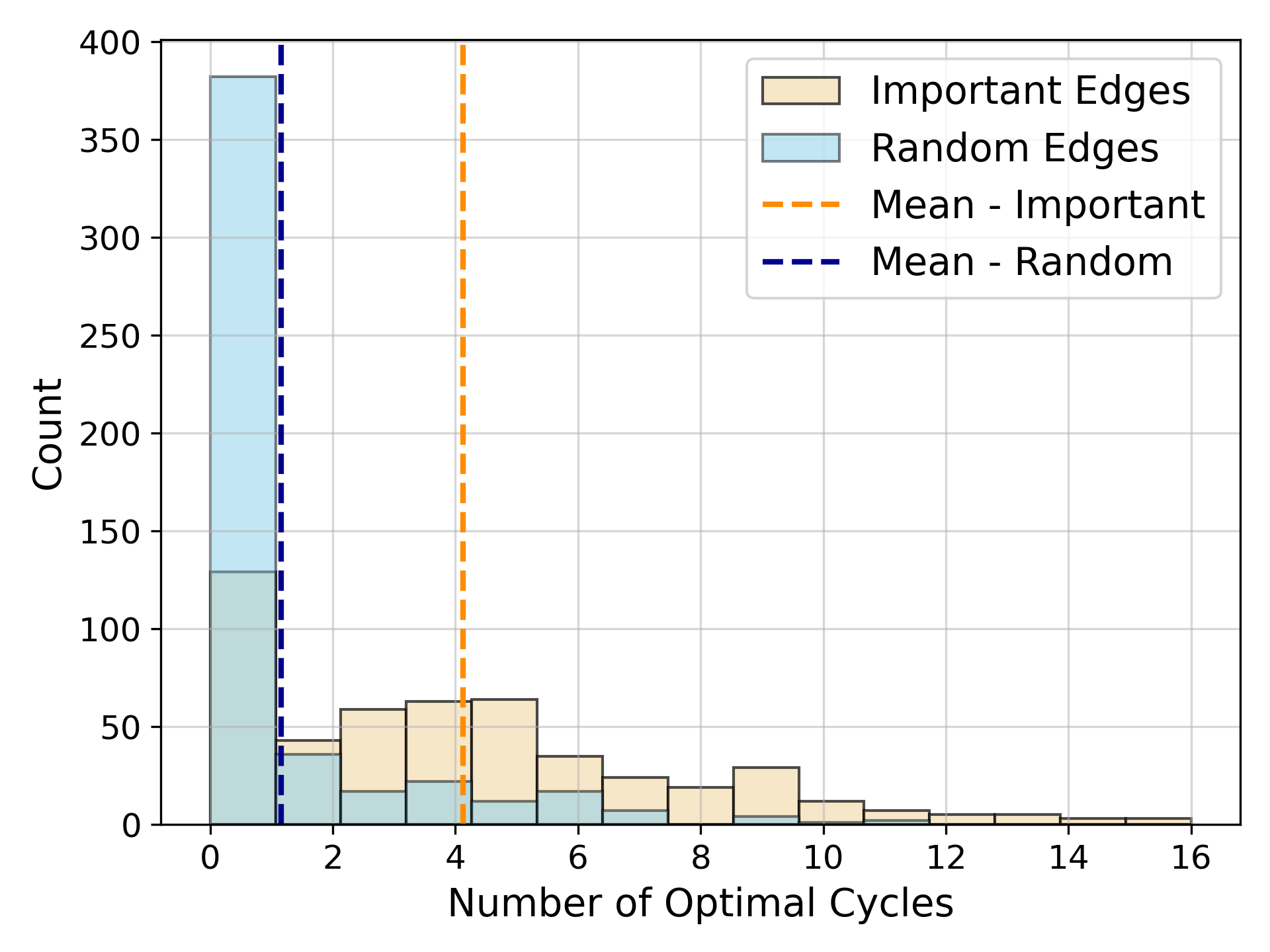}
    \end{minipage}
    \vspace{-10pt}
    \caption{(a) Distribution of distance to the prediction target (central node) of important atoms identified by our explanation vs. all atoms. (b)Distribution of the number of cycles involved in important edges identified in our explanation vs. randomly selected edges.}
    \label{fig:explain}
\end{figure*}

We produce explanation using the method in Section~\ref{subsec:explain}. We first provide visualizations to qualitatively show our explanation, and then we do a quantitative evaluation by connecting of our explanation to the MRO hypothesis and the topological data analysis (TDA) to reveal more insights.

\subsection{Explanation Visualization and MRO}
We visualize our explanation for a randomly sampled node. We provide both the global and the local version. All atoms are plotted in their actual 3D coordinates. The global version of the explanation is presented in Figure~\ref{fig:global}, where we visualize the top 50 important edges. The local version is in Appendix Figure~\ref{fig:1501_explanations}, where we zoom in to the top 10 closet nodes. From the visualizations, we see that edges close to the central node or far away from the central node may both be selected. Also, as we can see from the local version of the explanation, few edges within the top 10 closet edges is selected by our explanation. 

Material scientists proposed an MRO hypothesis, which basically says the atoms that lie in a range of medium distance (5 - 10 \text{\AA}) from the central atom play a more important role in determining its MG properties~\cite{ma2009power, sheng2006atomic, nomoto2021medium, egami2023medium}. We aggregate our selected edges to nodes (atoms) and plot the atom importance against their distance to the central node in Figure~\ref{fig:explain} (a). We see there are two modes of the important atoms, one for the closest ones and one for the medium-range ones, which matches the MRO hypothesis.

\subsection{Edge Importance vs. TDA}

\begin{table}[t]
\centering
\caption{Average number of optimal cycles participated by edges with high, medium, and low importance according to the explanation and randomly selected edges.}
\vspace{10pt}
\resizebox{\columnwidth}{!}{
\begin{tabular}{lcccc} 
\midrule
Edge Importance & High & Medium & Low & Random \\ 
\midrule
Avg \# Optimal Cycles & 4.130 & 1.202 & 0.874 & 1.148 \\ 
\midrule
\end{tabular}
}
\label{tab:num_cycles}
\end{table}

We perform TDA to further understand our edge importance explanations and see if the results recover meaningful topological structures.

\textbf{Persistent Homology (PH)}
PH is a widely used TDA method, a field of study that applies concepts from algebraic topology to data analysis~\cite{barannikov1994framed, zomorodian2004computing, edelsbrunner2008persistent}.  PH examines how topological features such as connected components, holes, or voids, emerge and disappear as one moves through different scales in a dataset. The persistence of certain topological features across scales can reveal important insights about the underlying structure of the data, making PH a powerful tool for material science. This method has also been applied to metallic glass to uncover important topological structures \cite{sorensen2020revealing}.
In PH, the concepts of ``birth'' and ``death'' are the essential quantities we would like to study, which visually represent the lifespan of topological features in a dataset. ``Birth'' refers to the scale at which a feature, like a connected component or a hole, first appears during the filtration process, while ``death'' denotes the scale at which this feature disappears or merges. 



\textbf{Explanation Results}
In our case, we apply PH to study the emergence and death of 1D hole as we increase the radius of a ball surrounding each atom. We perform the inverse analysis to pair the hole with a representative optimal cycle~\cite{obayashi2018volume}. In this way, each edge in the graph can be associated with a sequence of births and deaths of the cycles that it has participated in. We perform statistical analysis to see if there is significant difference between selected edges by our explanation and other edges. We plot and compare the distribution of the number of optimal cycles involved in the highest importance edges selected by our explanation versus randomly selected edges over multiple central nodes that are being explained in Figure~\ref{fig:explain} (b). We found that on average the importance edges participates in much more cycles compared to others, and there is a clear trend in the decrease in cycle number as the importance of edges decrease. The mean of the number of cycles involved in by edges in the four different group can be found in Table~\ref{tab:num_cycles}.




\section{Conclusion}\label{sec:conclusion}
In this paper, we study the connection between the local atomic structures of MGs and their EBs of the energy landscape. We formalize this problem as node regression on graphs and propose \model to solve the problem by effectively capturing the invariance of orthogonal transformations of the graph. We compare \model with several baseline models and demonstrate that \model performs the best. In addition, we extend the GNNExplainer to regression tasks and generate explanations. We further investigate the explanations with MRO and PH. We show a strong correlation between the importance of edge and the number of optimal cycles they involved in. Our work enables effective prediction and interpretation of MG EBs, bolstering material science research.


\section*{Acknowledgements}
This work was partially supported by NSF 2211557, NSF 1937599, NSF 2119643, NSF 2303037, NSF 2312501, NASA, SRC JUMP 2.0 Center, Amazon Research Awards, and Snapchat Gifts.

\section*{Impact Statement}\label{sec:impact}


This paper presents work whose goal is to advance the field of Machine Learning, especially on the area AI for material science. There are many potential societal consequences of our work, none of which we feel must be specifically highlighted here.

\nocite{langley00}

\bibliography{reference}
\bibliographystyle{icml2024}

\newpage
\appendix
\onecolumn
\section{Appendix} \label{app}
\subsection{The EB Prediction Problem}

In this section, we present a succinct high-level motivation for our work. Material scientists aim to use atomic structures of materials to predict their properties, such as ductility. However, this direct prediction is a challenging task. An alternative approach is to use an easier-to-predict intermediate quantity, e.g. the energy barriers, as a stepping stone. In other words, there is a shift from the paradigm 1:
$$\textbf{structures} -{predict} \to \textbf{properties}$$
to paradigm 2:
$$\textbf{structures} -predict \to \textbf{energy barriers} -analyze \to \textbf{properties}$$
The second paradigm has shown promising results, which is the focus of this work.

\subsection{Proof of Lemma~\ref{lemma:non-rotation}} \label{proof:non-rotation}
In this section, we prove the non-rotation decomposition Lemma~\ref{lemma:non-rotation} stated in Section~\ref{subsec:symmetrization}.

\begin{proof}
    From Lemma~\ref{lemma:decompose}, we know any non-rotation $\tilde{R} \in P\cdot \mathrm{SO}(3)$ has $\det{\tilde{R}} = -1$. By linearity of $\tilde{R}$ we know that
    \begin{align*}
        \langle -\tilde{R}(\vx^1), -\tilde{R}(\vx^2) \rangle &= \langle \tilde{R}(-\vx^1), \tilde{R}(-\vx^2) \rangle = \langle -\vx^1, -\vx^2 \rangle = \langle \vx^1, \vx^2 \rangle
    \end{align*}
    which shows hat $-\tilde{R}$ is also orthogonal according to Definition~\ref{def:ortho_trans}. As we know that $-\tilde{R}$ will have $\det{-\tilde{R}} = (-1)^3\cdot -1 = 1$, by Lemma~\ref{lemma:decompose} again we know that $-\tilde{R}$ is a rotation. By Euler Theorem~\ref{thm:euler}, we know there exists $[\alpha, \beta, \gamma] \in [0, 2\pi]^3$ such that $-\tilde{R} = O_{x_1}(\alpha)O_{x_2}(\beta)O_{x_3}(\gamma)$, which implies that $\tilde{R} = -O_{x_1}(\alpha)O_{x_2}(\beta)O_{x_3}(\gamma)$.
\end{proof}

\subsection{Proof of Theorem~\ref{thm:invariance}} \label{proof:invariance}
In this section, we prove the invariance in expectation Theorem~\ref{thm:invariance} stated in Section~\ref{subsec:symgnn}.

\begin{theorem}
    Assume $T_1, \dots, T_k$ are random transformations that follow a uniform distribution over all $T \in \OO(3)$. Then, $Sym$ is $\OO(3)$-invariant in expectation in the sense that $\mathbb{E}_{T_1, \dots, T_k} [Sym_{T_1, \dots, T_k}(G, \mZ, \mX)] = \mathbb{E}_{T_1, \dots, T_k}[Sym_{T_1, \dots, T_k}(G, \mZ, T_0\mX)]$ for any $T_0 \in \OO(3)$.
\end{theorem}

\begin{proof}
    From Lemma~\ref{lemma:ortho}, we know that all $T \in \OO(3)$ have the form of $T = (-1)^{\lambda} O_{x_1}(\alpha)O_{x_2}(\beta)O_{x_3}(\gamma)$ for $\lambda \in \{0, 1\}$ and $[\alpha, \beta, \gamma] \in [-\pi, \pi]^3$. A uniform distribution over all $T \in \OO(3)$ implies $\lambda \sim \mathrm{Bern}(0.5)$ and $[\alpha, \beta, \gamma] \sim \mathrm{Unif}([-\pi, \pi]^3)$.
    
    Now consider a specific orthogonal transformation $T_0 = (-1)^{\lambda_0} O_{x_1}(\alpha_0)O_{x_2}(\beta_0)O_{x_3}(\gamma_0)$. Then its composition with the uniformed distributed $T$ is
    \begin{align*}
        T \circ T_0 &= (-1)^{\lambda + \lambda_0}O_{x_1}(\alpha + \alpha_0)O_{x_2}(\beta + \beta_0)O_{x_3}(\gamma + \gamma_0)
    \end{align*}
    By the Bernoulli assumption, we get $(-1)^{\lambda} \sim (-1)^{\lambda + \lambda_0}$ as they both follow a discrete distribution on $\{-1, 1\}$ with probability $0.5$ of each value. Moreover, $\alpha \sim \mathrm{Unif}[-\pi, \pi]$ implies $\alpha + \alpha_0 \sim \mathrm{Unif}[-\pi + \alpha_0, \pi + \alpha_0]$. Since this only shifts the interval that supports the uniform distribution, the joint distribution of $(\cos(\alpha), \sin(\alpha)) \sim (\cos(\alpha + \alpha_0), \sin(\alpha + \alpha_0))$ due to periodicity, which further implies 
    $O_{x_1}(\alpha) \sim O_{x_1}(\alpha + \alpha_0)$. Similarly for $O_{x_2}(\beta)$ and $O_{x_3}(\gamma)$. Given the random variables (matrices) $(-1)^{\lambda}, O_{x_1}(\alpha), O_{x_2}(\beta)$, and $O_{x_3}(\gamma)$ are independent, we concluded that $T \sim T \circ T_0$.

    Now consider $Sym_{T_1, \dots, T_k}(G, \mZ, \mX) = \frac{1}{k} \sum_{i=1}^{k} Enc(G, \mZ, T_{i}(\mX))$. For each $i$, $T_i \sim T_{i} \circ T_0$ implies $T_i (\mX) \sim T_{i} \circ~T_0 (\mX)$, and thus $Enc(G, \mZ, T_{i}(\mX)) \sim Enc(G, \mZ, T_{i} \circ T_0(\mX))$ by transformation of random variables~\cite{billingsley2017probability}. Therefore, we also get $\mathbb{E}_{T_i}[Enc(G, \mZ, T_{i}(\mX))] = \mathbb{E}_{T_i}[Enc(G, \mZ, T_{i} \circ T_0(\mX))]$. Finally, by linearity of expectation, $\mathbb{E}_{T_1, \dots, T_k} [Sym_{T_1, \dots, T_k}(G, \mZ, \mX)] = \mathbb{E}_{T_1, \dots, T_k}[Sym_{T_1, \dots, T_k}(G, \mZ, T_0\mX)]$.
\end{proof}

\subsection{Analysis of expressiveness of SymGNN}\label{app:expressiveness}
In this section, we show by examples that SymGNN can capture more complex interactions between molecules compared to SchNet like methods. Consider two three molecule systems that have the following configurations where the atom type for node 2 is two whereas the atom type for node 1 and node 3 is one: 
\begin{enumerate}[label=(\alph*)]
    \item 1: $(-2,0,0)$, 2: $(0, 0, 0)$, 3:$(2, 0, 0)$ 
    \item 1: $(-\sqrt{2}, \sqrt{2}, 0)$, 2: $(0,0,0)$, 3:$(\sqrt{2}, \sqrt{2}, 0)$
\end{enumerate}
Notice that the edge distance between node 2 to each of node 1 and node 3 are two in both of these configurations, so SchNet cannot distinguish these two configurations based on the embedding of node 2, as it only takes into account the distance information. However, we shall see that SymGNN can differentiate these configurations based on node 2's embedding as it considers also higher-order information. To ease the computation, we consider a minimal setting of SymGNN where the encoder is the identity function, and each time two orthogonal transformations will be applied to the graph and then aggregated. Finally, node 2's embedding is calculated by simply summing up the message passed between node 2 and node 1 and between node 2 and node 3. Suppose the two orthogonal transformations sampled are a counterclockwise rotation in $xy$-plane by 45 degrees and a reflection around $y$-axis. It can be calculated that SymGNN will give the embedding $(0, 0, 0)$ for node two in the first configuration, but the second configuration gives $(2, 2+2\sqrt{2}, 0)$. Similar example can be given to show that SymGNN can also detect configurations that have equivalent angle structure, and thus we know SymGNN truly considers higher level information compared to those invariant methods that are based on scalerization.

\subsection{A Detailed Dataset Construction Process} \label{app:data_collect}
We employ molecular dynamics to simulate the behavior of a representative Cu64Zr36 metallic glass (MG) subjected to shear deformation. The simulated MG system comprises 8000 atoms, generated through the conventional melting-quenching procedure with varied cooling rates spanning from $10^{14}$ to $10^{10}$ $K/s$. To evaluate the influence of system size, we also simulate small system (i.e., 3000 atoms). To initiate the simulation, the sample is initially melted at 2000K under zero pressure for 1ns, facilitating the erasure of its initial configuration memory. Temperature and pressure control are maintained through the isothermal-isobaric (NPT) ensemble, employing a Nosé-Hoover thermostat~\cite{nose1984unified, hoover1985canonical}. Subsequently, the liquified state is rapidly quenched to 1K, with cooling rates ranging from $10^{14}$ to $10^{10}$ $K/s$. The resulting glassy structure is further relaxed to its local energy minimum through energy minimization, utilizing the conjugate gradient algorithm. The interatomic interactions within the system are described using the embedded-atom method (EAM) potential~\cite{mendelev2009development}. To ensure the statistical robustness of our findings, 9 independent metallic glass samples are generated for each cooling rate. A timestep of 1fs is adopted for all simulations, and the entire set of simulations is carried out using the LAMMPS package~\cite{plimpton1995fast}.

To obtain the energy barriers of atoms, we employ the activation-relaxation technique nouveau (ARTn)~\cite{barkema1996event, cances2009some} to calculate the energy barriers within MGs. Specifically, starting from a local energy minimum in the landscape, initial perturbations are introduced to a chosen atom and its nearest neighbors. This perturbation allows exploration along a direction of negative curvature, increasing the likelihood of locating a saddle point in the energy landscape. The Lanczos algorithm~\cite{barkema1996event} is then applied to guide the system to the saddle point by following the direction of negative curvature. A force tolerance of $0.05 eV/\text{\AA}^{-2}$ is chosen to ensure convergence of the saddle points. In accordance with previous investigations~\cite{fan2014thermally, fan2017energy, xu2018predicting}, 20 searches for saddle points are conducted for each atom. Consequently, the ARTn exploration focuses on determining the average energy barrier associated with atoms. This parameter is recognized as a key factor influencing the propensity for plastic rearrangement in disordered materials~\cite{tang2020bulk, tang2021energy}. The simulated raw dataset initially only contains nodes (atoms) along with their types and 3D coordinates. We construct edges between two nodes if their Euclidean distance is smaller than a threshold, which is chosen to be $5 \text{\AA} = 10^{-10}m$. 

\subsection{CuZr-Based MGs as Representative Examples}
In this section, we discuss choice of focusing on Cu-Zr based MGs in our dataset. First, Cu-Zr-based metallic glass is one of the most widely investigated MGs due to its outstanding mechanical properties and good glass-forming ability~\cite{cheng2011atomic}. Many well-known studies on MGs, such as those focusing on mechanical properties and ductility~\cite{liu2012microstructural, pauly2010transformation}, are centered on Cu-Zr. Additionally, Cu-Zr has been used as a standard MG example in ML research to study $\beta$ processes~\cite{wang2020predicting} and perform hierarchical structure analysis~\cite{hiraoka2016hierarchical}. Cu64Zr36 is known as the best glass former in this class of MGs and is commonly used as the archetype model in MD simulations~\cite{wang2020predicting}.  While other MGs are not included in this study, some common dynamic behaviors (e.g., relaxation, dynamical heterogeneity, shear band formation) are believed to be controlled by energy barriers, with structure-property relationships transferable between different MGs. 


\subsection{Abalation on Dataset Splits}
To check our method's robustness under different dataset configurations, we performed an ablation study with two more random dataset splits. These splits use different sets of randomly sampled training graphs, one randomly sampled graph for validation, and rest for testing. The result can be found on Table~\ref{tab:both_datasets}. We see that SymGNN consistently outperforms SchNet in all the dataset splits we considered.

\begin{table}[ht]
\centering
\begin{tabular}{lccc}
\toprule
& Model       & Train  & Test   \\
\midrule
Original    & SymGNN & 0.8368 & \textbf{0.7859} \\
            & SchNet & 0.7858 & 0.7588 \\
\midrule
New Split 1 & SymGNN & 0.8600 & \textbf{0.7613 }\\
            & SchNet & 0.7778 & 0.7583 \\
\midrule
New Split 2 & SymGNN & 0.8362 & \textbf{0.7645} \\
            & SchNet & 0.7070 & 0.6948 \\
\bottomrule
\end{tabular}
\caption{Performance comparison across original and new dataset splits.}
\label{tab:both_datasets}
\end{table}


\subsection{Ablation on Number of Orthogonal Transformations}
In our experiments, we found that the model performance is relatively stable with respect to the number of transformations. We performed an ablation study with different number of aggregated orthogonal transformations. We presented experiments results with 2, 4, 6, and 12 transformations. Although using 12 transformations led to out-of-memory (OOM) issues, we found that both 2 and 4 transformations yielded effective results, with the performance using 2 transformations even being slightly better than our reported results. We hypothesize that as long as we are sampling various orthogonal transformations from a reasonable distribution, the framework can benefit from the symmetrization module and achieve good results. The number of transformations mostly influences the convergence time rather than the performance. For example, using 2 transformations required 3000 epochs for convergence, while 4 transformations required 2300 epochs. This observation aligns with our expectations, as fewer transformations necessitate more iterations for the model to capture the necessary information from the data. The results can be found in Table~\ref{tab:number_ortho}. 

\begin{table}[h!]
\centering
\begin{tabular}{cccccc}
\toprule
Number of Transformations & 0      & 2      & 4      & 6      & 12 \\ 
\midrule
Train                     & 0.8736 & 0.8302 & 0.8072 & 0.8368 & OOM \\ 
Test                      & 0.2669 & 0.7901 & 0.7778 & 0.7858 & OOM \\ 
\bottomrule
\end{tabular}
\caption{Performance ablated on different number of orthogonal transformations}
\label{tab:number_ortho}
\end{table}

\subsection{Comparison with Data Augmentation}\label{subsec:data_aug}

Empirically, we find that our symmetrization module can learn a condensed subspace of the orthogonal transformations corresponding to the task, which allows more effective aggregations. In this section, we provide an analysis with the subspace learned by SymGNN. We report the mean and concentration for each distribution that controlled one Euler angle after the training. For rotations, we have
$$(\mu_i, \kappa_i) = (-0.3414, 0.2985), (0.7023, 0.9146), (-1.5622, 1.3543)$$ and $$(\mu_i, \kappa_i) = (-0.4102, 3.0350), (-0.4946, 3.0645), (-0.7043, 0.1533)$$ for the rest. We perform an estimation of how much volume we need in order to capture 80 percent of the whole probability density. We notice that Von Mises distribution with a bigger $\kappa$ is approximately a Gaussian distribution with variance $\frac{1}{\kappa}$, and a Von Mises distribution with a smaller $\kappa$ is approximately a uniform distribution. Therefore, we approximate the estimation using a similar density-volume estimation with four Gaussian random variables $\mathcal{N}(-1.5622, \frac{1}{1.3543}), \mathcal{N}(-0.4102, \frac{1}{3}), \mathcal{N}(-0.4946, \frac{1}{3}), \mathcal{N}(0.7023, \frac{1}{0.9146})$, and two uniform distributions between $[-\pi, \pi]$. We found that for these four Gaussian distributions, intervals of length $2.21, 3.37, 3.37, 4.01$ around the mean can approximately yield $94.5$ percent of density. Therefore, in total it will yield $0.945^4 \approx 0.8$ of density. By picking all the uniform distribution, we know that at least $80$ percent of the density can be included by less than $7$ percent of the density. This analysis implies the effectiveness of the learning. We should note that this estimation is rough and an overestimation, since in reality the distributions that are approximated as uniform are also more centered.

\subsection{Explanation on SchNet}

We compare explanations generated with different models to demonstrate the effectiveness of both our model and our explanations. We ran GNNExplainer on our strongest baseline, SchNet, and found both similarities and differences in the explanation outputs compared to ours. On one hand, the explanations are similar to the explanations generated with our model, with the most important edges being of mid-range distance to the central node (as in Figure 3). On the other hand, the comparison to TDA shows a clear difference. As we discussed in Section 6.1, we compare the edge importance identified by ML explanations to the edges captured by the TDA optimal cycles. A comparison between our explanations and SchNet explanations is shown in the table below. We observe that for SchNet explanations, there is no significant difference between the number of optimal-cycle edges participated in by those high importance edges and low importance edges (high vs. low = 1.312 vs. 1.24), whereas ours is very significant (high vs. low = 4.13 vs. 0.874). We hypothesize that this difference is due to the lack of expressivity in SchNet, as it utilizes only edge distance information and not any higher-order information, such as angles. The result can be found in Table~\ref{tab:avg_cycles_schnet}.

\begin{table}[h!]
\centering
\begin{tabular}{lcccc}
\toprule
Ave Number of Cycles & High & Medium & Low & Random \\ 
\midrule
SymGNN              & 4.130 & 1.202 & 0.874 & 1.148 \\
SchNet              & 1.312 & 1.036 & 1.240 & 1.074 \\
\bottomrule
\end{tabular}
\caption{Average number of cycles involved with high/medium/low impact edges in SchNet.}
\label{tab:avg_cycles_schnet}
\end{table}

\subsection{Explanation Visualization}
In this section, we provide more visualizations of the explanation. First, we show the pairing local explanation result for the node shown in Figure~\ref{fig:global} in Figure~\ref{fig:1501_explanations}. We note that none of the edges within the top 10 closest nodes are being selected as important edges by our explanation. Visualizations of more nodes are shown in Figure~\ref{fig:more_explanations}. In the figures, the left column presents the global version plot whereas the right column presents the local version.

\begin{figure}[htbp]
    \centering
    \begin{minipage}{0.44\textwidth}
        \centering
        \includegraphics[width=\linewidth]{figures/1501_global.png}
        \label{fig:1501_local}
    \end{minipage}
    \begin{minipage}{0.44\textwidth}
        \centering
        \includegraphics[width=\linewidth]{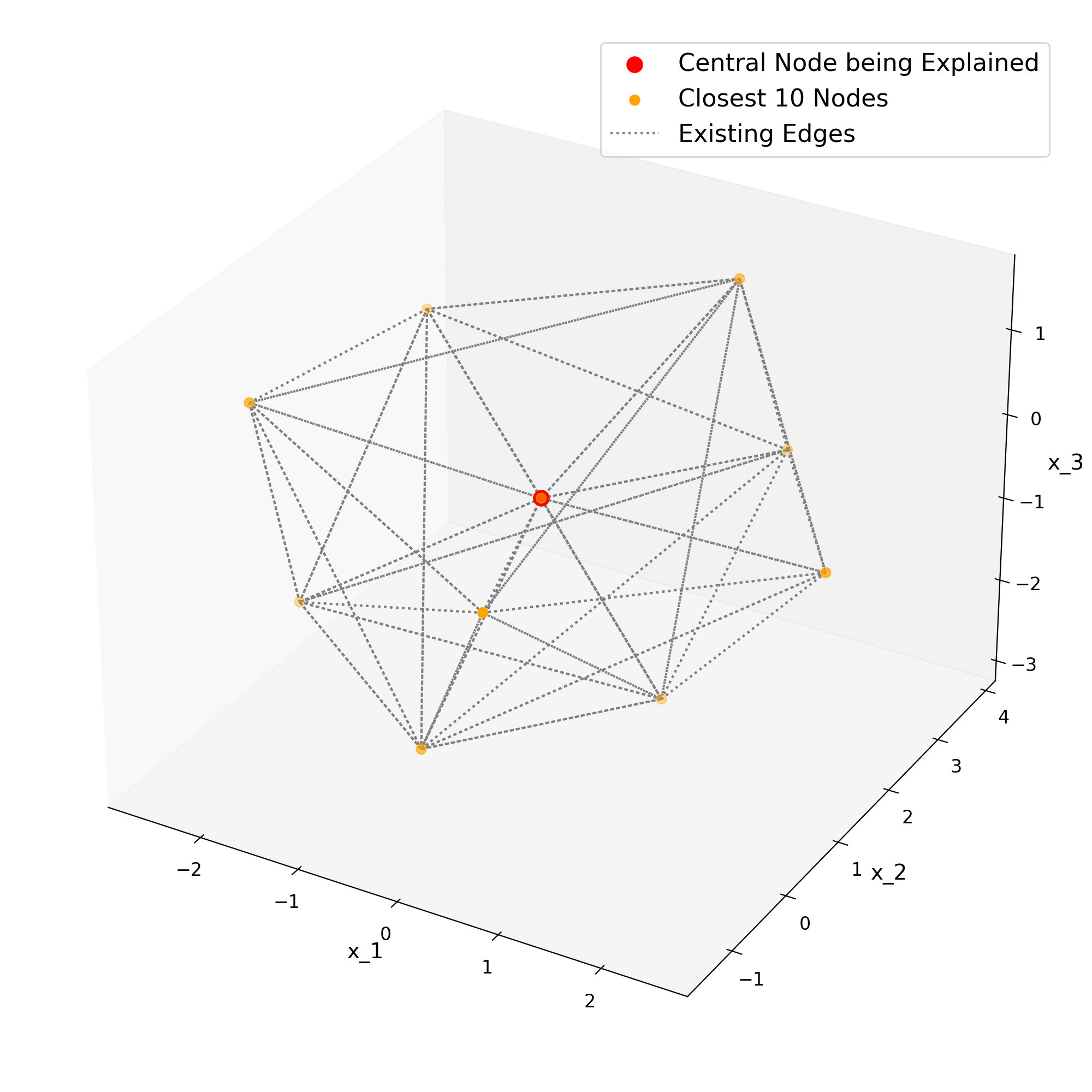}
        \label{fig:1501_global}
    \end{minipage}
    \caption{Local and global explanation visualizations node 1501.}
    \label{fig:1501_explanations}
\end{figure}

\begin{figure}[htbp]
    \centering
    \begin{minipage}{0.46\textwidth}
        \centering
        \includegraphics[width=\linewidth]{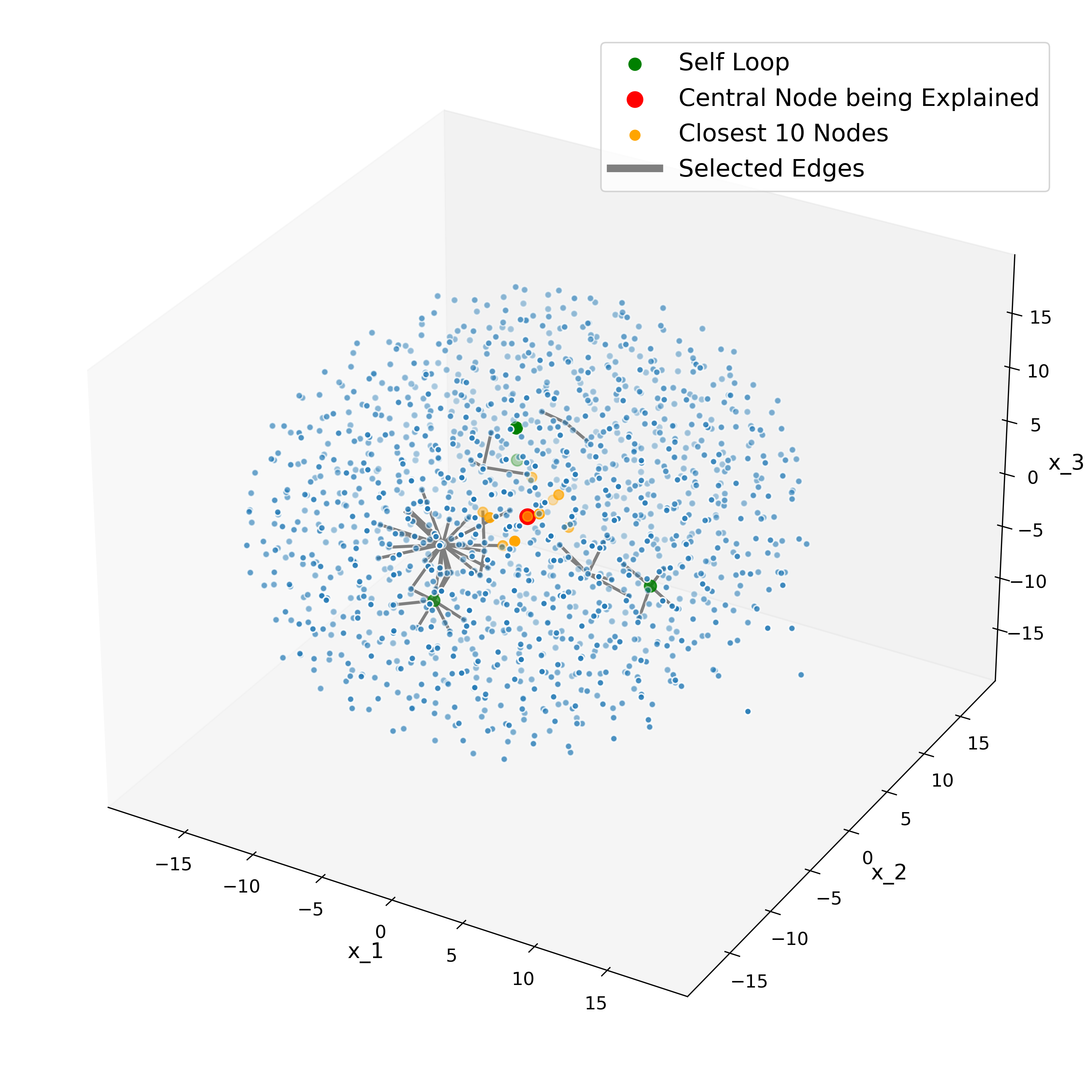}
        \label{fig:1502_global}
    \end{minipage}
    \begin{minipage}{0.46\textwidth}
        \centering
        \includegraphics[width=\linewidth]{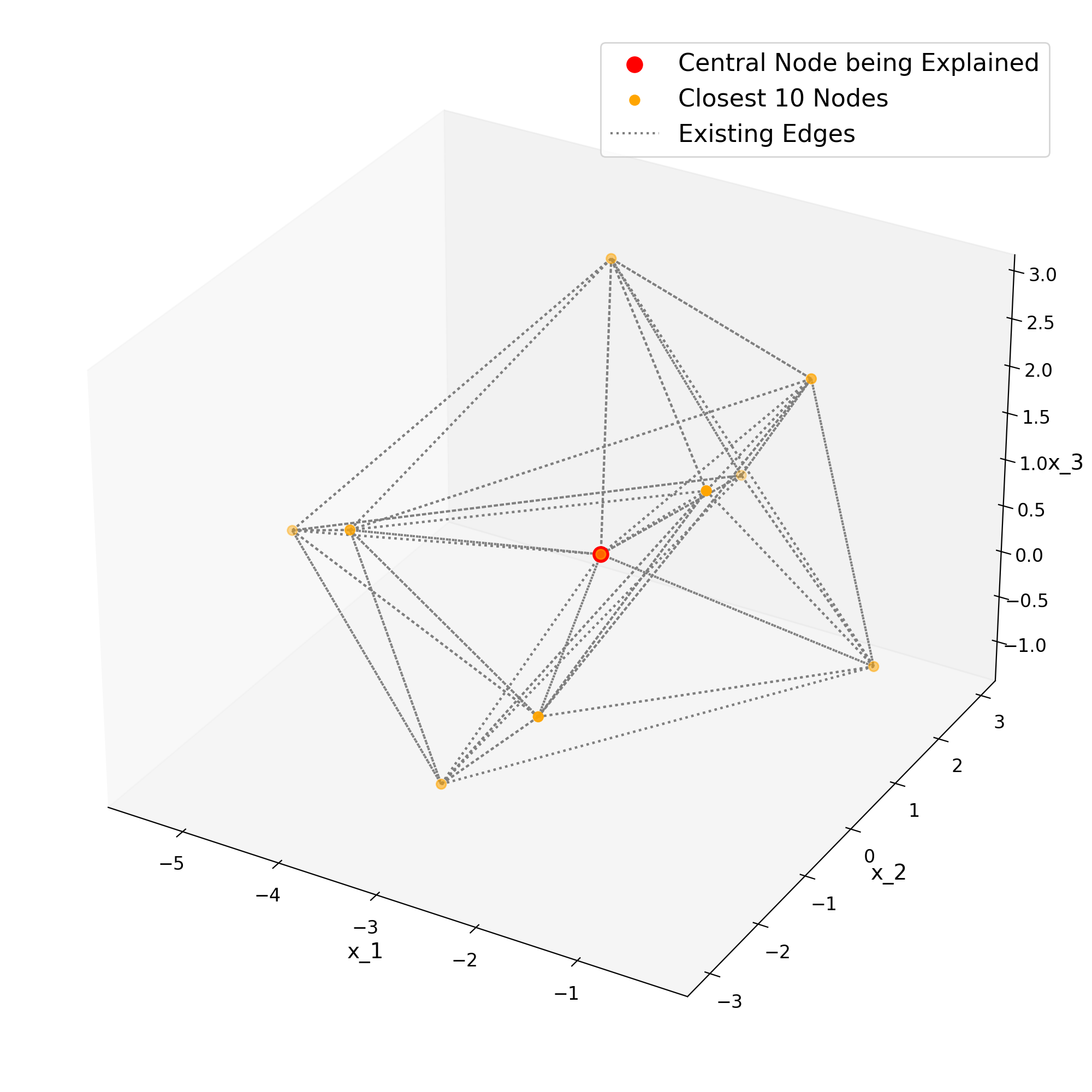}
        \label{fig:1502_local}
    \end{minipage}
    
    \vspace{-8mm}
    \begin{minipage}{0.46\textwidth}
        \centering
        \includegraphics[width=\linewidth]{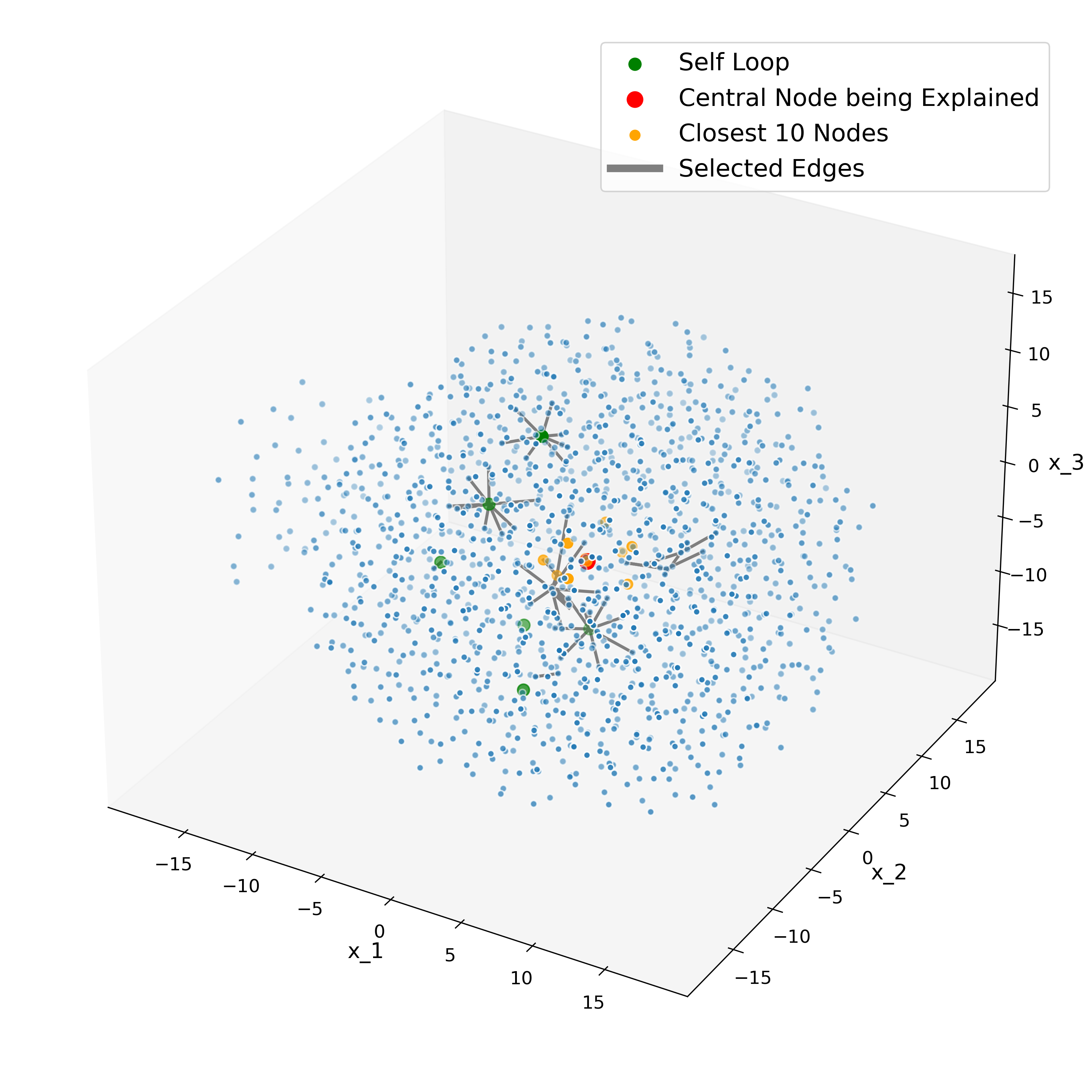}
        \label{fig:1468_global}
    \end{minipage}
    \begin{minipage}{0.46\textwidth}
        \centering
        \includegraphics[width=\linewidth]{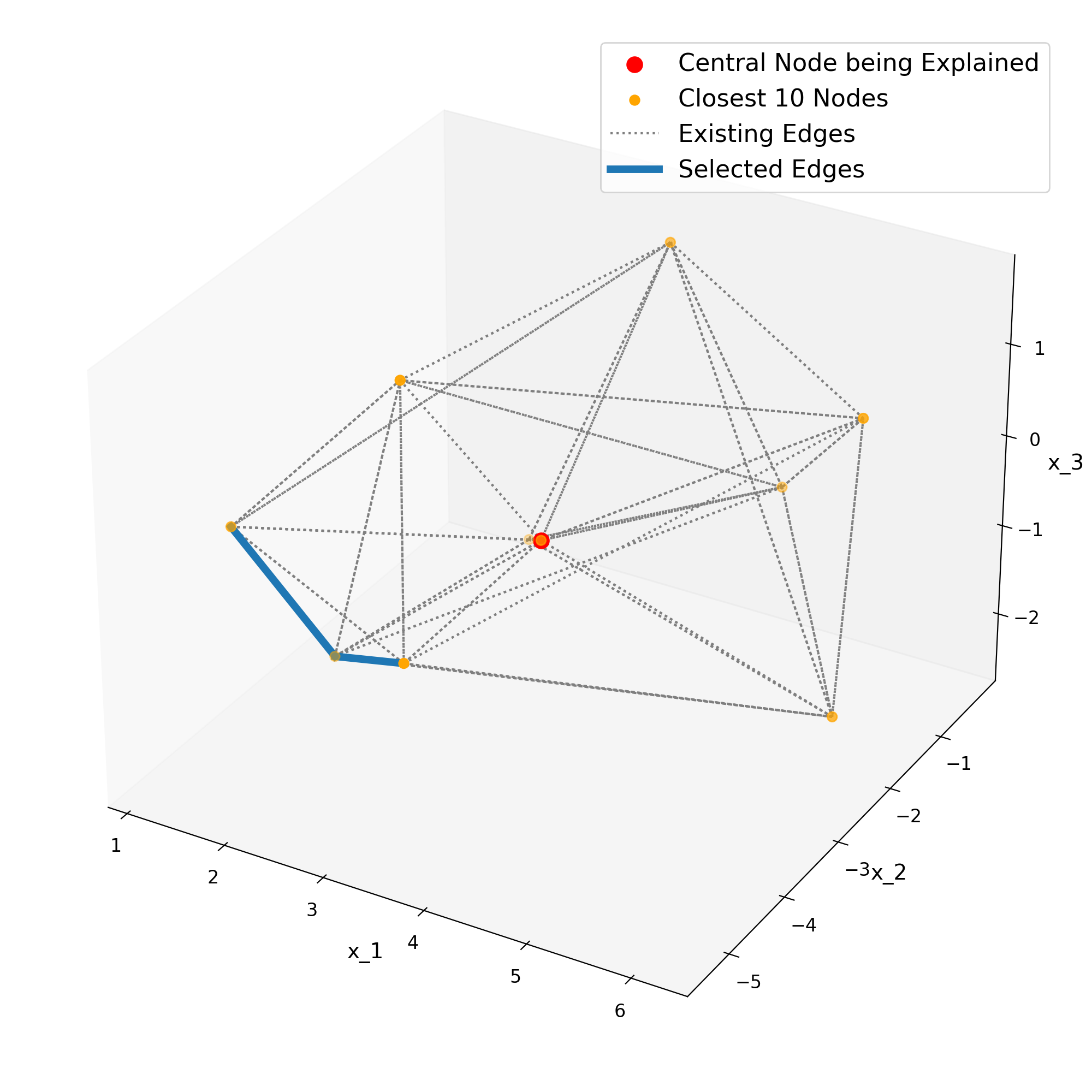}
        \label{fig:1468_local}
    \end{minipage}

    \vspace{-8mm}
    \begin{minipage}{0.46\textwidth}
        \centering
        \includegraphics[width=\linewidth]{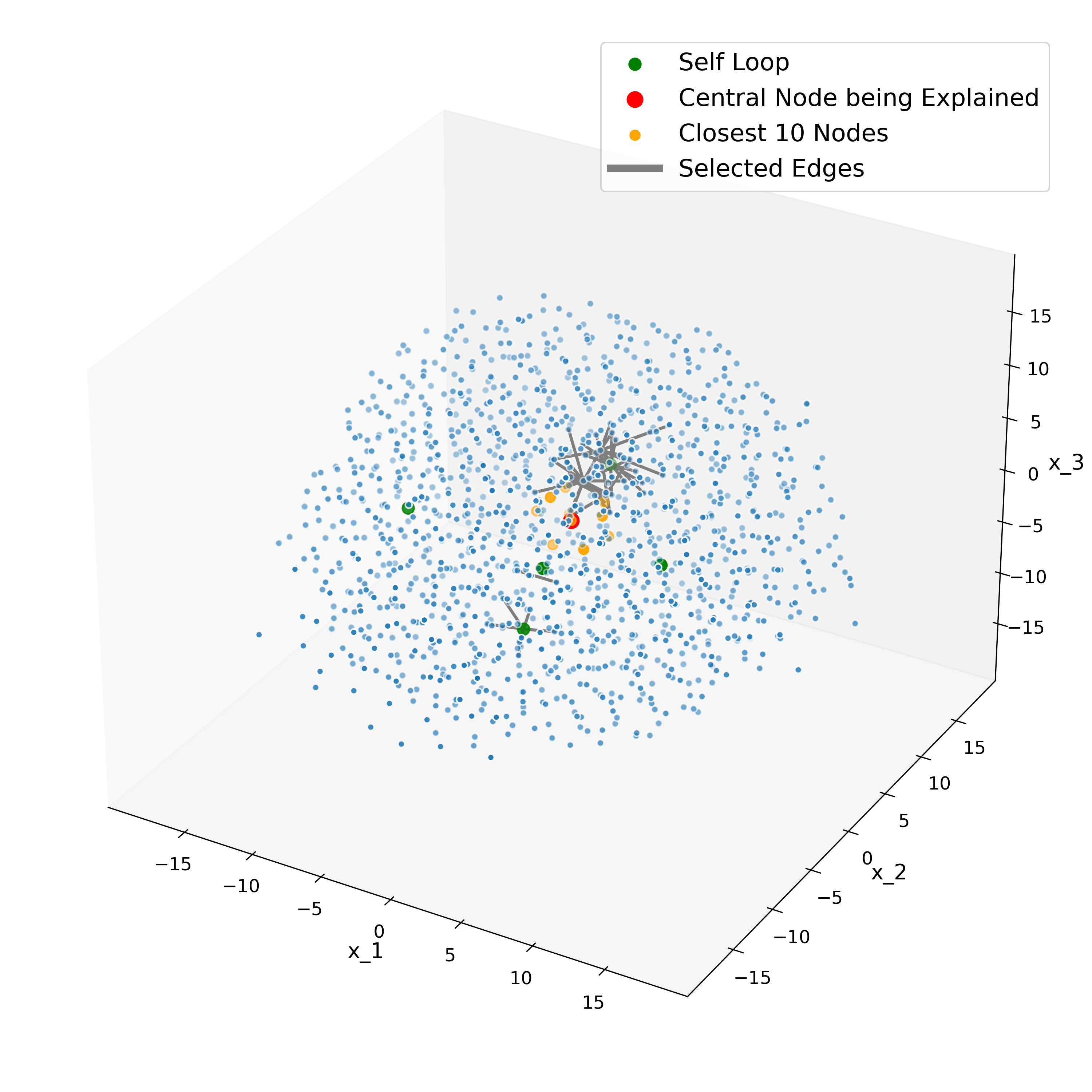}
        \label{fig:1534_local}
    \end{minipage}
    \begin{minipage}{0.46\textwidth}
        \centering
        \includegraphics[width=\linewidth]{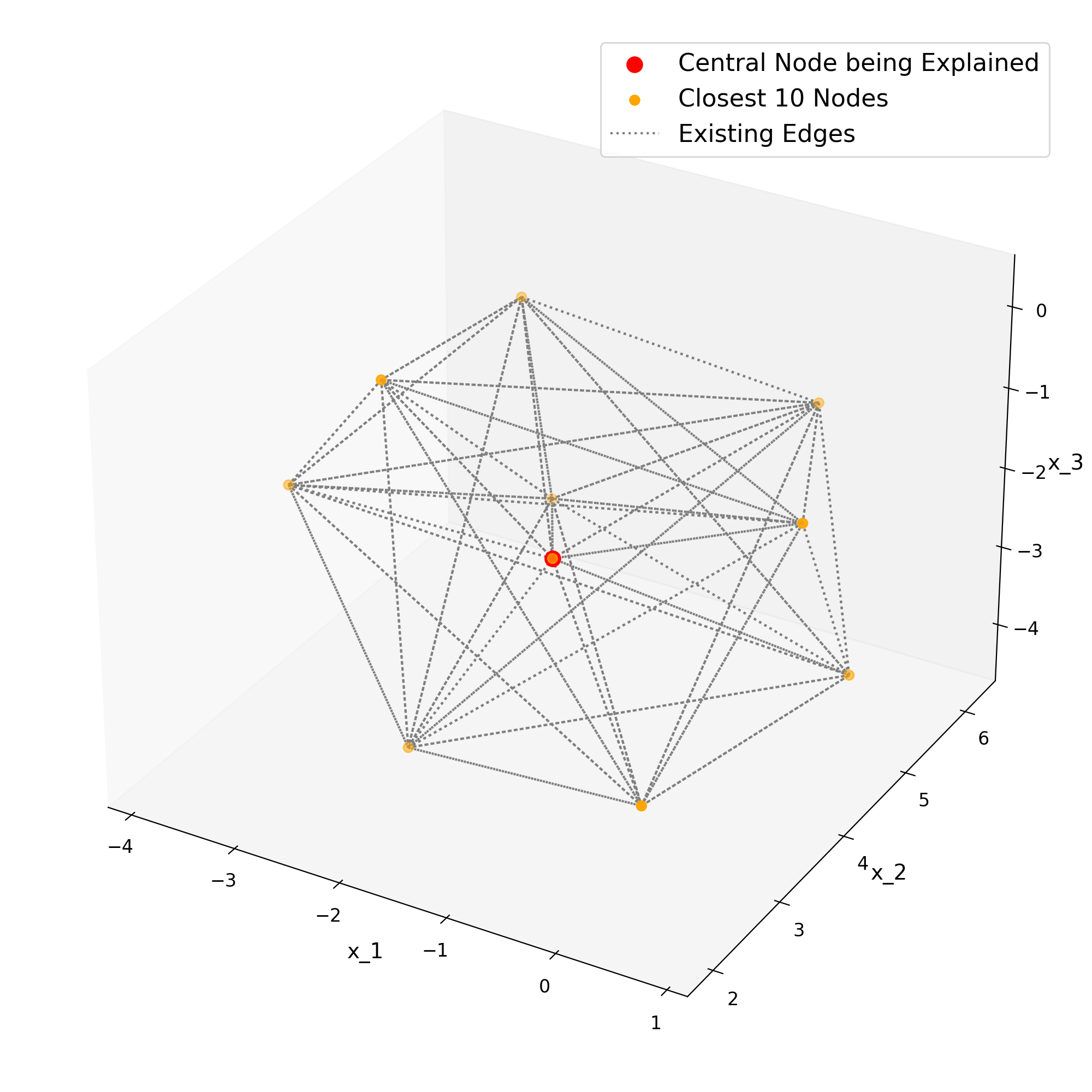}
        \label{fig:1534_global}
    \end{minipage}

    \vspace{-8mm}

    \caption{Local and global explanation visualizations for more nodes.}
    \label{fig:more_explanations}
\end{figure}




\end{document}